%% file: 0_main.tex
\documentclass{svjour3}     

\smartqed  
\usepackage{graphicx}
\usepackage[usenames, dvipsnames]{color}
\usepackage{mathptmx}      

\input{macros-packages}

\input{macros}

\input{macros-trees}

\begin{document}

\titlerunning{
 On Huang and Wong's algorithm for generalized binary split trees
}
\title{
 On Huang and Wong's algorithm for generalized binary split trees
\\[1em]\large \centering Marek Chrobak
	$\cdot$
  Mordecai Golin
  $\cdot$
	J.~Ian Munro
	$\cdot$
	Neal~E.~Young
  \thanks{
    To appear in Acta Informatica (2022).
    See~\cite{chrobak_etal_isaac_2015} for a conference version of the first result in this paper.
    See~\cite{chrobak2015optimal_erratum,chrobak_etal_simple_bcst_algorithm_2019,CHROBAK2021104707}
    for improved versions of other results in~\cite{chrobak_etal_isaac_2015}.
  }
}
\author{}

\institute{Marek Chrobak,
  {University of California at Riverside.\\
    \and Research supported by NSF grants CCF-1217314 and CCF-1536026.}\\
 Mordecai Golin,
	{Hong Kong University of Science and Technology.\\
          \and Research funded by HKUST/RGC grant FSGRF14EG28 and RGC CERG Grant 16208415.}\\
	J.~Ian Munro,
	{University of Waterloo.\\
	\and Research funded by NSERC and the Canada Research Chairs Programme.} \\
	Neal E.~Young,
	{University of California at Riverside.\\
          \and Research supported by NSF grant IIS-1619463.}
}

\date{}

\maketitle


\begin{abstract}
  Huang and Wong~\cite{StephenHuang1984} proposed a polynomial-time dynamic-programming algorithm
  for computing optimal generalized binary split trees.
  We show that their algorithm is incorrect. Thus, it remains open whether such trees
  can be computed in polynomial time.
  Spuler~\cite{Spuler1994Paper,Spuler1994Thesis}
  proposed modifying Huang and Wong's algorithm to obtain an algorithm for a different problem:
  computing optimal two-way-comparison search trees. 
  We show that the dynamic program underlying Spuler's algorithm is not valid,
  in that it does not satisfy the necessary optimal-substructure property
  and its proposed recurrence relation is incorrect. 
  It remains unknown whether the algorithm is guaranteed to compute a correct overall solution.
\end{abstract}


\section{Introduction}%
\label{sec: introduction}\label{sec: intro}
\input{1_introduction}


\section{Huang and Wong's \gbsplit algorithm is incorrect}%
\label{sec: huang flaw}%
\input{2_errors_huang_wong}


\section{A \twoWCST algorithm by Spuler fails on some subproblems}%
\label{sec: spuler flaw}%
\input{3_errors_spuler}


\bigskip
\paragraph{Note:}
We would like to use this opportunity to acknowledge yet another
error in the literature on binary split trees, this one in our
own paper~\cite{chrobak_etal_isaac_2015}. In that paper we introduced
a perturbation method that can be used to extend algorithms for
binary search trees with keys of distinct weights to instances where 
key-weights need not be distinct, and we claimed that this method can
be used to speed up the computation of optimal binary split trees
to achieve running time $O(n^4)$. 
(Recall that in binary split trees from~\cite{Huang1984,Perl1984,Hester1986},
the equality-test key in each node must be a most likely key among keys reaching the node.)
As it turns out, this claim is not valid. In essence, the perturbation 
approach from~\cite{chrobak_etal_isaac_2015} does not apply to binary
split trees because such perturbations affect the choice of the equality-test key
and thus also the validity of some trees.
{See~\cite{chrobak2015optimal_erratum,chrobak_etal_simple_bcst_algorithm_2019}
  for an erratum, full proofs of the remaining results,
  and pointers to follow-up work.}


\bigskip

\paragraph{Acknowledgements.}
We are very grateful to the anonymous reviewer who meticulously
verified our proofs and calculations, and provided numerous
comments that significantly improved the rigor and clarity of the paper.




\bibliographystyle{plainurl} 
\bibliography{bib}

\newpage
\appendix
\section*{APPENDIX}

\input{5_python-code_v3}

\end{document}

%% file: macros-packages.tex
\usepackage[x11names]{xcolor}

\usepackage{amsmath} 
\usepackage{amsfonts} 
\usepackage{amssymb} 
\usepackage{xspace}   
\usepackage{relsize} 
\usepackage{fancyvrb} 
\usepackage{enumerate} 
\usepackage{enumitem} 
\usepackage{centernot} 
\usepackage{adjustbox} 
\usepackage{multirow} 
\usepackage{stmaryrd} 
\usepackage{units} 
\usepackage[normalem]{ulem} 
\usepackage{ifthen} 
\usepackage{stmaryrd}

\usepackage[external]{forest}
\usetikzlibrary{backgrounds}

\usepackage{lineno}

\usepackage{listings}
\usepackage[hidelinks]{hyperref}           


%% file: macros.tex
\newtheorem{Claim}[lemma]{Claim}

\newcommand{\HuangLemma}[3]
{\addtocounter{lemma}{1}\par\medskip\noindent\textbf{Lemma #1 from~\cite{StephenHuang1984} (#2)} {\em #3\/} \smallskip\par}

\newcommand{\HuangLemmaTight}[3]
{\addtocounter{lemma}{1}\par\medskip\noindent\textbf{Lemma #1 from~\cite{StephenHuang1984} (#2)} {\em #3\/} \par}

\newcommand{\margincomment}[1]%
{{\marginpar{{\footnotesize\begin{minipage}{0.75in}%
          \begin{flushleft}%
            {#1}%
          \end{flushleft}%
        \end{minipage}%
      }}}\ignorespaces}

\newcommand{\ignore}[1]{}

\colorlet{darkgreen}{green!45!black}

\newcommand{\etal}{et al.\xspace}
\newcommand{\etals}{et al's\xspace}

\newcommand{\braced}[1]{{ \left\{ #1 \right\} }}
\newcommand{\angled}[1]{{ \left\langle #1 \right\rangle }}

\newcommand{\Keys}{{\sf K}}

\newcommand{\query}{v}

\newcommand{\myproblem}[1]{{\mbox{\sc #1}}\xspace}

\newcommand{\gbsplit}{\myproblem{gbsplit}}

\newcommand{\twoWCST}{\myproblem{2wcst}}

\newcommand\eps\varepsilon

\newcommand{\barT}{{\bar{T}}}

\newcommand{\opt}[1]{\operatorname{\sf opt}({#1})}
\newcommand{\optstar}[1]{\operatorname{\sf opt}^\ast({#1})}
\newcommand{\hwTree}{\operatorname{\tau}}

\newcommand{\compnode}[2]{\ensuremath{\angled{ \query \mathbin{#1} #2 }}}

\newcommand{\KEQ}{e}
\newcommand{\KLT}{s}

\newcommand{\hLT}{h_{\mathsmaller 1}}
\newcommand{\hGE}{h_{\mathsmaller 2}}

\newcommand{\prob}[2]{{#1}({#2})}

\newcommand{\weight}[1]{{\prob p {#1}}}

\newcommand{\nofposvalues}[1]{{#1_{\scriptscriptstyle{+}}}}

\newenvironment{mycases}[1]
{ \begin{list}{}{%
      \setlength \labelsep {0.5em}
      \setlength \labelwidth {0pt}
      \setlength \listparindent {\parindent}
      \setlength \leftmargin {#1}
      \setlength \rightmargin {0pt}
      \setlength \itemsep {0.7ex}
      \setlength \parskip {2pt}
      \setlength \parsep {1pt}   
      \setlength \labelwidth {-0.5em}
      \setlength \topsep {0pt}
      \setlength \partopsep {0pt}
    }}
{ \end{list}                  } 

\newenvironment{outercases}
{ \begin{mycases}{0em} }
{ \end{mycases}                  }


%% file: macros-trees.tex
\newcommand{\key}[1]{{\mathsf{#1}}}
\newcommand{\keyA}{{\key A}}
\newcommand{\keyB}{{\key B}}
\newcommand{\keyC}{{\key C}}
\newcommand{\keyD}{{\key D}}
\newcommand{\keyE}{{\key E}}
\newcommand{\keyF}{{\key F}}

\newcommand{\keyAzero}{{\key{A0}}}
\newcommand{\keyAone}{{\key{A1}}}
\newcommand{\keyAtwo}{{\key{A2}}}
\newcommand{\keyAthree}{{\key{A3}}}
\newcommand{\keyBzero}{{\key{B0}}}
\newcommand{\keyBfour}{{\key{B4}}}
\newcommand{\keyCzero}{{\key{C0}}}
\newcommand{\keyDzero}{{\key{D0}}}
\newcommand{\keyDone}{{\key{D1}}}
\newcommand{\keyEzero}{{\key{E0}}}

\newcommand{\keyFone}{{\key{F1}}}
\newcommand{\keyFtwo}{{\key{F2}}}

\newcommand{\keyNzero}{{\key{N0}}}

\newcommand{\keyPzero}{{\key{P0}}}

\newcommand{\keyQzero}{{\key{Q0}}}
\newcommand{\keyQone}{{\key{Q1}}}
\newcommand{\keyRzero}{{\key{R0}}}

\newcommand{\keySzero}{{\key{S0}}}
\newcommand{\keySone}{{\key{S1}}}
\newcommand{\keyTzero}{{\key{T0}}}
\newcommand{\keyTtwo}{{\key{T2}}}
\newcommand{\keyUzero}{{\key{U0}}}
\newcommand{\keyUone}{{\key{U1}}}
\newcommand{\keyVzero}{{\key{V0}}}
\newcommand{\keyVthree}{{\key{V3}}}
\newcommand{\keyWzero}{{\key{W0}}}
\newcommand{\keyWone}{{\key{W1}}}
\newcommand{\keyXzero}{{\key{X0}}}
\newcommand{\keyXtwo}{{\key{X2}}}
\newcommand{\keyYzero}{{\key{Y0}}}
\newcommand{\keyYone}{{\key{Y1}}}
\newcommand{\keyZzero}{{\key{Z0}}}

\newcommand{\treesize}{\footnotesize}

\newenvironment{yntree}{%
  \treesize 
  \forest
  for tree={
    every leaf node={leaf},
    every interior node={interior},
  },
  where n={1}{where n'={2}{edge label={node[edgeYes]}}{}}{},
  where n={2}{where n'={1}{edge label={node[edgeNo]}}{}}{},
}{\endforest}

\newenvironment{yntrees}{%
  \treesize 
  \forest
  where level=1{below=2ex}{}, 
  for descendants={
    where level=1{}{
      every leaf node={leaf},
      every interior node={interior},
      where n={1}{where n'={2}{edge label={node[edgeYes]}}{}}{},
      where n={2}{where n'={1}{edge label={node[edgeNo]}}{}}{},
    }
  },
}{\endforest}

\newcommand{\leafcolor}{red!3}
\newcommand{\interiorcolor}{green!4}
\newcommand{\subtreecolor}{black!4}
\newcommand{\edgecolor}{black!66}

\newcommand{\subtreeparams}[2]{
  \forestset{
    subtree1/.style={
      anchor=north, 
      child anchor=north, 
      outer sep=0pt,
      rounded corners=0.7em, 
      shape=semicircle,
      inner sep=1pt,
      text depth=4pt,
      xscale=#1,
      content format={
        \noexpand\scalebox{#2}[1]{\forestoption{content}}
      },
    },
  }
  \forestset{
    subtree/.style={
      subtree1,
      draw={\edgecolor, thin, dotted}, 
      edge={\edgecolor, thin}, 
      fill=\subtreecolor, 
    },
  }
  \forestset{
    dimtree/.style={
      subtree1,
      draw={gray, thin, dotted},
      edge={gray, thin, dotted}, 
      fill=gray!3,
      text=black!50, 
    },
  }
  \forestset{
    greentree/.style={
      subtree1,
      draw={\edgecolor, thin}, 
      edge={\edgecolor, thin}, 
      fill=\interiorcolor, 
    },
  }
}
\newcommand{\defaultsubtreeparams}{
  \subtreeparams{0.91}{1.1}
}
\defaultsubtreeparams 

\forestset{
  interior/.style = {
    draw={\edgecolor, thin}, 
    edge={\edgecolor, thin}, 
    rounded rectangle, 
    inner sep=1.8pt, 
    text height=1.6ex, 
    text depth=.5ex, 
    fill=\interiorcolor,
    anchor=north, 
  },
  leaf/.style = {
    draw={\edgecolor, thin}, 
    edge={\edgecolor, thin}, 
    rectangle, 
    inner sep=1.7pt, 
    text height=1.7ex, 
    text depth=.3ex, 
    child anchor=north, 
    fill=\leafcolor,
    anchor=north, 
  }, 
  nodraw/.style = {draw=none, fill=none, anchor=north}, 
  equalitytest/.style = {calign=last}, 
  every leaf node/.style={
    if n children=0{#1}{}
  },
  every interior node/.style={
    if n children=0{}{#1}
  },
}

\tikzstyle{edgeCommon} = [circle, font=\tiny, inner sep=0pt, outer sep=1.5pt] 
\tikzstyle{edgeYes} = [edgeCommon, text depth=0ex, text height=1.2ex, auto=right, node contents={y}, pos=0.35]
\tikzstyle{edgeNo} = [edgeCommon, text depth=0ex, text height=1.2ex, auto=left, node contents={n}, pos=0.4, outer sep=1pt] 
\tikzstyle{edgeZero} = [edgeCommon, auto=right, node contents={0}, pos=0.7] 
\tikzstyle{edgeOne} = [edgeCommon, auto=left, node contents={1}, pos=0.7]
\tikzstyle{edgeLess} = [edgeCommon, auto=right, node contents={$<$}]
\tikzstyle{edgeGreater} = [edgeCommon, auto=left, node contents={$>$}, pos=0.5]
\tikzstyle{edgeEqual} = [edgeCommon, auto=right, node contents={$=$}, pos=0.4]

\tikzset{label distance=-1pt}

\forestset{
  bright/.style={
    fill=yellow!12,
    draw={black, thick}, 
    edge=thick,
  },
  dim/.style={
    fill=gray!3,
    draw={gray, thin, dotted}, 
     text=black!50, 
  },
  dimedge/.style={
    edge={thin, dotted}
  },
  gbst/.style={
    l sep=0ex, 
    s sep=1em, 
    draw={\edgecolor, thin}, 
    edge={\edgecolor, thin}, 
    rectangle, 
    rounded corners=.6em, 
    fill=brown!15, 
    inner sep=.21em, 
    outer sep=0pt,
    content format={
      \noexpand\ensuremath{\noexpand\begin{array}{@{}c@{}}%
                 \forestoption{content}%
                 \noexpand\end{array}\noexpand}%
             },
           }
}

\newcommand{\lab}[1]{label={-90:{\footnotesize\ensuremath{{#1}}}}}



%% file: 1_introduction.tex


Given an ordered set $\Keys$ of $n$ keys, a \emph{generalized binary split tree} $T$
is a form of binary search tree where each node $N$ has two associated keys in $\Keys$:
an \emph{equality-test} key and a \emph{split} key~\cite{StephenHuang1984}.
For any query $\query\in\Keys$, a \emph{search} for $\query$ in $T$ starts at the root.
If $\query$ equals the root's equality-test key, then the search halts.
Otherwise, the search recurses in the left or right subtree,
depending on whether or not $\query$ is less than the root's split key.
A correct tree $T$ must have $n$ nodes, and the search for each query $\query\in \Keys$
must halt at the node whose equality-test key is $\query$.
(There must be exactly one such node for each $\query\in \Keys$.)
Given also a probability distribution $p$ on $\Keys$,
the \emph{cost} of a tree $T$ is the expected number of nodes visited
when searching in $T$ for a random query $\query$ drawn from $p$.
The goal, given $\Keys$ and $p$, is to compute a tree $T$ of minimum cost
(thus minimizing, over any tree $T$ of this form, the expected number of two-way comparisons made when searching in $T$).
We denote this problem \gbsplit. See Fig.~\ref{fig: intro bsplit} for an example.
{Following Huang and Wong, here we focus on the so-called \emph{successful-queries} variant,
  in which all queries are guaranteed to be in $\Keys$.
  (In the \emph{general} variant, arbitrary queries are allowed.)}


\begin{figure}\centering
   \noindent%
  \input{TREE_intro_bsplit.tex}%

   \caption{The picture on the left shows an example of a generalized binary split tree
     for key interval $\{\keyA,\keyB,\keyC,\keyD,\keyE,\keyF\}$. 
    Each node is labeled with its equality key and its probability, as well as
	the node's split key (except that split keys are omitted at leaves, where they are irrelevant).
	The total cost of this tree is $0.3\cdot 1 + 2\cdot (0.2\cdot 2) + 3\cdot (0.1\cdot 3) = 2$.
	In all figures in the paper we use a more compact representation,
	shown on the right, where split keys are omitted. (Each node's split key can be any key that separates
     the equality keys in the left subtree from those in the right subtree.)
  }\label{fig: intro bsplit}
\end{figure}


Huang and Wong~\cite{StephenHuang1984} proposed a polynomial-time algorithm for \gbsplit.
We show (in Theorem~\ref{thm: huang flaw}, Sect.~\ref{sec: huang flaw})
that their algorithm and claimed proof of correctness are wrong.
The reason is that their dynamic program does not satisfy the claimed optimal-substructure property.
Consequently, as far as we know, it is not known whether \gbsplit has a polynomial-time algorithm.

A closely related problem is to find an optimal \emph{two-way comparison} search tree,
in which each node is associated with just one key and one binary comparison operator---{equality or less-than}. We use \twoWCST to denote this problem.
{ (See Fig.~\ref{fig: spuler example} for an example.)}
Spuler~\cite{Spuler1994Paper,Spuler1994Thesis} proposed several \twoWCST algorithms.
He described two of his proposed \twoWCST algorithms
(for the successful-queries and general variants, respectively)
as ``straightforward'' modifications of Huang and Wong's \gbsplit algorithm, but
he gave no formal proof of correctness,
explaining only that correctness follows from the dynamic-programming formulation,
in particular from the underlying recurrence relation.

We show (Theorem~\ref{thm: spuler flaw}, Sect.~\ref{sec: spuler flaw}) that this recurrence relation is wrong,
and his algorithm computes incorrect solutions to some subproblems in the dynamic program.
Here also the dynamic program does not satisfy the assumed optimal-substructure property.
This counter-example is only for a subproblem, not a full instance,
so the overall correctness of his proposed algorithm remains open.
(Here also we focus on the successful-queries variant only.)


\paragraph{Historical context.}
The study of optimal binary search trees began with \emph{three-way comparison} search trees.
These have only one key  associated with each node,
and comparing the given query to that key has \emph{three} possible
  outcomes---less than, equal to, or greater than.
Knuth's classical dynamic-programming algorithm computes a minimum-cost tree of this kind
(supporting both successful and unsuccessful queries) in time $O(n^2)$~\cite{Knuth1971}.

Following Knuth's suggestion~\cite[\S6.2.2~ex. 33]{Knuth1998},
various authors began exploring trees based on two-way (binary) comparisons.
Sheil~\cite{Sheil1978} introduced \emph{median split trees}---generalized binary split trees
where the split key at each node $N$ must be a median
key among the set $\Keys_N$ of keys whose search visits node $N$,
and the equality-test key must be a most likely key among $\Keys_N$.
He gave an $O(n\log n)$-time algorithm to compute a median split tree (for the successful-queries variant).
Other authors~\cite{Huang1984,Perl1984,Hester1986} then introduced \emph{binary split trees}---generalized binary split trees with the added restriction
that the equality-test key at each node must be a most likely key
among keys reaching the node. These trees can be thought of as
a relaxation of median split trees, without the restriction that the split key has to be a median key.
Their algorithms compute minimum-cost binary split trees in $O(n^5)$ time
for both the successful-queries and general variants.
(See also the note at the end of this paper.) Huang and Wong~\cite{StephenHuang1984}
then introduced \gbsplit (generalized binary split trees) as defined above,
and proposed an $O(n^5)$-time algorithm for the problem, the one we show here to be incorrect.

Subsequently, the algorithm was extended by Chen and Liu to \myproblem{multiway} \gbsplit,
a variant of \gbsplit that requires multiple split keys per node~\cite{chen_optimal_1991}.
Chen and Liu's algorithm and proof of correctness are directly patterned on Huang and Wong's.
Their proof is invalid (and we believe their algorithm to be incorrect)
for the same reason that Huang and Wong's proof and algorithm fail.
(See the remark at the end of Sect.~\ref{sec: huang flaw}.)

As mentioned above, Spuler~\cite{Spuler1994Paper,Spuler1994Thesis} proposed several \twoWCST algorithms
without proof of correctness. Anderson \etal~\cite{Anderson2002}
gave the first proof that \twoWCST is in polynomial time. Their algorithm runs in time $O(n^4)$
and is restricted to the successful-queries variant. Chrobak \etal~\cite{chrobak_etal_isaac_2015,chrobak2015optimal_erratum,chrobak_etal_simple_bcst_algorithm_2019}
gave a somewhat simpler $O(n^4)$-time  algorithm for the  general variant.

Beyond pointing out errors in the literature on binary search trees,
we hope that the constructions underlying our counter-examples will contribute
to a better understanding of the difficulties involved in designing algorithms for
\gbsplit and \twoWCST, leading to better algorithms or even new hardness results.


%% file: TREE_intro_bsplit.tex
{
\scriptsize
 
\begin{forest}
  for descendants={gbst, l sep=1ex},
  [, phantom, 
  [{\normalsize~}, nodraw [, phantom [, phantom [, phantom ]]]]
  [{=\,} \keyC\\0.3\\<\keyE,
          [{=\,} \keyD\\0.2\\<\keyB, 
              [{=\,} \keyA\\0.1]
              [{=\,} \keyB\\0.1]
          ]
          [{=\,} \keyF\\0.2\\<\keyE,
          [, phantom]
          [{=\,} \keyE\\0.1]
          ]
 ]
[, phantom [, phantom [, phantom [, phantom [, phantom][, phantom]]]]]
[, phantom [, phantom [, phantom [, phantom [, phantom][, phantom]]]]]
[{\normalsize~}, nodraw [, phantom [, phantom [, phantom [, phantom][, phantom]]]]]
[\keyC\\0.3,
          [\keyD\\0.2, 
              [ \keyA\\0.1]
              [ \keyB\\0.1]
          ]
          [\keyF\\0.2,
          [, phantom]
          [\keyE\\0.1]
          ]
 ]
]
\end{forest}
}


%% file: 2_errors_huang_wong.tex


\newcommand{\hfigtree}[2]{\ensuremath{T_{\ref{fig: #1}#2}}}

This section gives our first main result:
a proof that Huang and Wong's proposed \gbsplit algorithm~\cite{StephenHuang1984}
has a fundamental flaw.

\begin{theorem}\label{thm: huang flaw}%
  Huang and Wong's \gbsplit algorithm~\cite{StephenHuang1984} is incorrect.
  There is a \gbsplit instance $(\Keys, p)$ for which it returns a non-optimal tree.
\end{theorem}

\renewcommand{\int}[2]{[#1,#2]}

We summarize their algorithm and analysis, give the intuition behind the failure, then prove the theorem.
The basic intuition is that, for the dynamic program that Huang and Wong define,
the optimal-substructure property fails.
The proof gives a specific counter-example and verifies it.
  The counter-example can also be verified computationally
  by running the Python code for Huang and Wong's algorithm in Appendix~\ref{sec: code huang}.

Fix any \gbsplit instance $(\Keys, p)$.
Assume without loss of generality that the keys are $\Keys=\{1,2,\ldots,n\}$.
Regarding the probability vector $p$, for convenience, throughout the paper we
drop the constraint that the probabilities must sum to 1, and we use 
``probabilities'' and ``weights'' synonymously,
allowing their values to be arbitrary non-negative reals.
(To represent probabilities, these values can be appropriately normalized.)

During a search, the outcome of each less-than comparison narrows the current search interval,
  while the outcome of each (failed) equality test removes one key within the interval
  from consideration. Thus, at each node in any search tree, the set of keys reaching the node
  consists of some interval of keys, minus some so-called \emph{holes}---keys removed from consideration by previous equality tests.
  Next we formally define an exponentially large (!)
  class of subproblems that arise in this way,
  along with a natural recurrence relation for their cost. We then discuss how Huang and Wong attempt
  to reduce the number of subproblems to $O(n^3)$.

Abusing notation, a \emph{query interval} $I=\int i j$ is the set of contiguous keys $\{i,i+1,\ldots,j\}$.
Given any query interval $I$ and any subset $H\subseteq I$ of ``hole'' keys,
consider the subproblem $(I,H)$ formed by the subset of keys $I\setminus H$, with the 
weight distribution obtained from $p$ by restricting to $I\setminus H$.
Let $\opt {I,H}$ denote the minimum cost of any generalized binary split tree for this subproblem.
Let $p(I\setminus H) = \sum_{k\in I\setminus H} p_k$ denote the total weight of its keys. 

If $H = I$ then the subproblem can be handled by an ``empty'' tree, so $\opt {I,H} = 0$.
Otherwise, letting $I=\int i j$, the definition of generalized binary split trees gives the recurrence
\begin{linenomath*}
\begin{equation*}
  \opt {I , H} ~=~ 
  \weight{I\setminus H}+
  \displaystyle\min_{\substack{\KLT\in [i, j+1];\\ \KEQ\in I \setminus H}}
  ~\big(~\opt {\int i {s-1}, \,H_{\KEQ} \cap \int i {s-1} }
  \,+\,   \opt {\int s j,\, H_{\KEQ} \cap \int s j}~\big)
\end{equation*}
\end{linenomath*}

\noindent 
where $H_{\KEQ} = H \cup \{\KEQ\}$. (Here $\KLT \in [i, j+1]$ ranges over the possible split keys;\footnote
{Huang and Wong allow $n+1$ as a split key,
  which is inconsistent with their stated definition of \gbsplit.
  This is a minor technicality---any tree that uses $n+1$ as a split key
  is easily converted into an equally good tree that does not.}
$\KEQ \in I \setminus H$ ranges over the possible equality keys.)

\smallskip

The goal is to compute $\opt {\Keys,\emptyset}$. The recurrence above allows arbitrary equality keys $\KEQ$,
so it gives rise to exponentially many hole sets $H$,
resulting in a dynamic program with exponentially many subproblems.
Huang and Wong propose a dynamic program with $O(n^3)$ subproblems $(I,h)$,
one for each interval $I$ and integer $h\le |I|$.  Specifically, they define
\begin{linenomath*}
\begin{equation*}
  \optstar {I, h} ~=~ \min\{\, \opt {I, H} : H\subseteq I,\, |H|=h\,\},
\end{equation*}
\end{linenomath*}
which is the minimum cost of any tree for interval $I$ minus \emph{any} hole set of size $h$.
Each such tree will have $|I|-h$ nodes. (Their paper uses ``$p[i-1,j,h]$'' to denote $\optstar {\int i j,h}$.) 
We refer to any such subproblem $(I, h)$ as an \emph{HW-subproblem}.

They develop a recurrence for $\optstar{I, h}$ as follows. For any node $N$ in an optimal tree,
define $N$'s \emph{interval} $I_N$ and \emph{hole set} $H_N$ in the natural way
so that interval $I_N$ contains those key values that, if searched for in $T$ with the equality tests 
ignored, would reach $N$,
and $H_N\subseteq I_N$ contains those keys \emph{in interval $I_N$} that are equality keys at ancestors of $N$.
Hence, the set of keys reaching $N$ is $I_N\setminus H_N$, and the subtree rooted at $N$ is a solution 
for the subproblem $(I_N, H_N)$, as well as the HW-subproblem $(I_N, |H_N|)$,
which we refer to as \emph{the HW-subproblem arising at $N$}.  Huang and Wong's Lemma 1 states:


\HuangLemma{1}{ambiguous}
{``Subtrees of an optimal generalized binary split tree
  are optimal generalized binary split trees.''}


This statement is ambiguous in that it doesn't specify \emph{for which subproblem} the subtree is optimal.
Consider any subtree $T'$ of an optimal tree $T^*$. Let $T'$ have root $N$, interval $I_N$ and hole set $H_N$.
The first interpretation of their Lemma~1 is that $T'$ must be an optimal solution for $({I_N, H_N})$.
With this interpretation (following the first recurrence above), the lemma is indeed true.
But another interpretation is that $T'$ must be an optimal solution for the HW-subproblem
$({I_N, |H_N|})$ arising at $N$. This interpretation is not the same---the HW-subproblem specifies only the \emph{number} of holes,
and choosing different holes can give a cheaper tree,
so it can be that $\optstar{I_N, |H_N|} < \opt{I_N, H_N}$. As we shall see below,
it is the second interpretation that  underlies the recurrence relation that Huang and Wong propose,
but, with that interpretation, as our Theorem~\ref{thm: no optimal substructure} shows,
the above lemma is false because the HW-subproblems do not have optimal substructure.

The ambiguity in Lemma 1 appears to be their first misstep.
They follow it with the following (correct) observation:


\HuangLemma{2}{correct}{
  Let $N$ be the root of a subtree $T'$ with interval $I$ in an optimal generalized binary split tree $T^*$.
  The equality-test key $e_N$ of $N$ must be the least frequent key
  among those in $N$'s interval $I_N$ that do not occur (as an equality-test key) in 
  the left and right subtrees of $N$.\footnote 
  {To avoid confusion, note that the lemma does not preclude
    a descendant $D$ of $N$ from having an equality-test key $e_{D}$ 
    that is more likely than $e_{N}$, because $e_{N}$ might not be in $D$'s interval.
    So it does not imply that the equality-test key $e_{N}$ at $N$ 
    is as likely as all equality-test keys in the subtree rooted at $N$.
    For example, see keys $\keyAtwo$ and $\keyDone$ in tree $T_{2a}$ in Fig.~\ref{fig: huang small}.}
}


\begin{linenomath*}
\begin{proof}
  The proof is a simple exchange argument.
  Suppose for contradiction that $e_N$ is more likely than some key $k$ in $I_N$
  and $k$ does not occur as an equality-test key in   the left and right subtrees of $N$.
  Then $k$ is a hole at $N$, so it must be the equality-test key $k=e_{N'}$ of some ancestor $N'$ of $N$.
  A contradiction is obtained by observing that exchanging $e_N$ and $e_{N'}$ gives a correct tree cheaper than $T^*$.
  \hfill\qed
  \vspace*{-1em}

  {}~{}
\end{proof}

Huang and Wong's Lemma~2 above (with the second, incorrect interpretation of their Lemma 1) suggests the following idea.
To find a hopefully optimal tree $\hwTree(I,h)$ for the HW-subproblem $(I, h)$,
consider each possible root split key and each possible split of the $h$ hole slots.  For each,
first find  optimal left and right subtrees for their respective subproblems,
and then take the equality key at the root to be the least-likely key in $I$
that is not an equality test in either subtree. Among trees obtained in this way,
take $\hwTree(I,h)$ to be one of minimum cost. Following this idea, their algorithm
(as detailed on pages 118--120 of their paper)
solves any given HW-subproblem $(I,h)$, where $I=[i,j]$ is non-empty, as follows:
\end{linenomath*}

\begin{center}
  \begin{minipage}{4.5in}
    \hrule
    \medskip
    \renewcommand{\labelenumii}{\theenumii}
    \renewcommand{\theenumii}{\theenumi.\arabic{enumii}.}

    \begin{enumerate}
    \item For each triple $(\KLT, \hLT, \hGE)$ where $\KLT\in [i, j+1]$ (the split key),
      and $\hLT$ and $\hGE$  (the numbers of holes in the left and right subtrees)
      are non-negative integers
      such that
      $\hLT + \hGE = h+1$,
      $\hLT \le |[i, s-1]| = \KLT-i$,
      and $\hGE \le |[s, j]| = j-\KLT+1$,
      construct one possible \emph{candidate tree} $T(\KLT, \hLT, \hGE)$ as follows:

      \smallskip 
      
      \begin{enumerate}[label*=\arabic*.]

      \item Give $T(\KLT, \hLT, \hGE)$ left and right subtrees $\hwTree([i, \KLT-1], \hLT)$ and $\hwTree([\KLT, j], \hGE)$.

        \smallskip 

      \item\label{step: e} 
 	Give the root of $T(\KLT, \hLT, \hGE)$ split key $\KLT$ and equality-test key $\KEQ$, 
    \emph{where $\KEQ$ is a least-likely key in $I$  that is not an equality-test key in either subtree.}

      \end{enumerate}
      \smallskip

    \item
    	Among trees $T(\KLT, \hLT, \hGE)$ so constructed, take $\hwTree(I,h)$ to be one of minimum cost.
    \end{enumerate}
    \hrule 
  \end{minipage}
\end{center}

{}~{}

The algorithm is not hard to implement. Appendix~\ref{sec: code huang} gives Python code for it (30 lines).

Note that, by their Lemma~2, the choice for $e$ in Line~\ref{step: e} would be correct
\emph{if} the second interpretation of their Lemma~1 was correct. We surmise that this line of thinking led 
Huang and Wong to their algorithm.

To justify the algorithm, Huang and Wong proceed as follows.  Fix any
execution of the algorithm
(breaking ties arbitrarily; see the remarks below).
For any HW-subproblem $(I,h)$ that it solves, let $\optstar{I, h}$ denote the minimum cost of any tree for the 
subproblem. Recall that $\hwTree(I, h)$ denotes the algorithm's solution (tree) for the subproblem,
presumably of cost $\optstar{I, h}$.
Huang and Wong first state a correct base case:


\HuangLemmaTight{3}{correct}{%
  \[\optstar{\emptyset, 0} = 0.\]%
}%
\noindent %
But their Lemma~4 then claims that, for any  non-empty interval $I=\int i j$ and any number of holes $h \le |I|$,
the following recurrence relation holds:
%


\HuangLemma{4}{incorrect}{
  \begin{linenomath*}
    \[\displaystyle\optstar{I, h} 
      \,=\, \min_{\KLT,\hLT}~\big(~ p(T(\KLT, \hLT, \hGE))
      \,+\, \optstar{ \int i {\KLT-1}, \hLT } \,+\, \optstar{ \int \KLT j, \hGE} ~\big) \]
  \end{linenomath*}
  where the minimum is over all legal combinations of $\KLT$ and $\hLT$,
  and
  $\hGE = h-\hLT+1$, and
  $p(T(\KLT, \hLT, \hGE))$ is the weight of keys in the tree $T(\KLT, \hLT, \hGE)$
  as defined above.}


\paragraph{Ambiguities in Lemma 4.}
During the execution of the algorithm, in Steps 1.2 and 2, ties may arise in choosing a minimizer.
Different choices can lead to different subtrees 
for any given $T(\KLT, \hLT, \hGE)$, with different hole sets.
Huang and Wong do not explicitly discuss tie-breaking,
and in the absence of such a rule
$p(T(\KLT, \hLT, \hGE))$ is not uniquely determined by the
subproblem $(I, h)$ and the parameters $(\KLT, \hLT, \hGE)$.
But the refutations we give here hold no matter how ties are broken. 

More significantly, our statement of their Lemma 4 corrects what we believe is an error.
Namely, their statement of the lemma
has ``$w(I, h)$'' where we have ``$p(T(\KLT, \hLT, \hGE))$'',
with $w(I, h)$ (on their page 118) defined as the ``total weight of the optimal GBST for''
the HW-subproblem $(I, h)$.
We believe that they had in mind the recurrence as we give it
(using $p(T(\KLT, \hLT, \hGE))$),
mainly because this recurrence
is the one that their algorithm, as defined on pages 118--120 of their paper,  actually uses.

Our Theorem~\ref{thm: no optimal substructure}, next,
refutes their Lemma~4 regardless of this issue---it refutes any recurrence based on the class of HW-subproblems $\{(I, h)\}$,
by showing that the class doesn't have the optimal-substructure property.
In Theorem~\ref{thm: huang flaw} and elsewhere,
by ``Huang and Wong's algorithm'',
we mean the algorithm as defined in pages 118--120 of their paper (independently of their statement of Lemma~4).
Our refutation of that algorithm, after Theorem~\ref{thm: no optimal substructure} below,
gives an instance on which it fails.


  \begin{theorem}\label{thm: no optimal substructure}
    There exists a \gbsplit instance $(\Keys,p)$ with the following property.
    In every optimal tree $T^*$ for $(\Keys,p)$, there is at least one node $N$ such that
    the subtree $T^*_N$  rooted at $N$ in $T^*$
    is not optimal for the HW-subproblem $(I_N,|H_N|)$ arising at $N$.
    (The tree $T^*_N$ has cost strictly larger than $\optstar{I_N, |H_N|}$.)
  \end{theorem}
 

\begin{proof}
Before we describe $(\Keys,p)$, we first describe an HW-subproblem for which using a minimum-cost tree $T'$
can be a bad choice globally. The HW-subproblem is $(I_9, 2)$, with $h=2$ holes
and interval $I_9$ consisting of nine keys
$I_9 =
\{\keyAone$, $\keyAtwo$, $\keyAthree$, $ \allowbreak\keyBzero$, $\keyBfour$,
$\keyCzero$, $\keyDzero$, $\keyDone$, $\keyEzero\}$,
ordered lexicographically, with weights as follows:

\begin{quote}\centering\small

  \begin{tabular}{|r||c|c|c|c|c|c|c|c|c|}\hline
    key    & $\keyAone$ &$\keyAtwo$ & $\keyAthree$ & $\keyBzero$ 
    & $\keyBfour$ & $\keyCzero$ & $\keyDzero$ & $\keyDone$ & $\keyEzero$ \\ \hline
    weight & 20 & 20 & 20 & 10 & 20 & 5 & 10 & 22 & 10 \\ \hline
  \end{tabular}

\end{quote}


\begin{figure}\centering
 \noindent%
  \input{TREE_error_d=1.tex}%

  \caption{Subtrees $\hfigtree{huang small}{a}$ and $\hfigtree{huang small}{b}$ for 9-key interval $I_9$ with $h=2$. 
    { $\hfigtree{huang small}{a}$ is missing the two keys   $\keyAthree, \keyBfour$; 
      $\hfigtree{huang small}{b}$ is missing  $\keyAthree, \keyDone$.} 
    Each node shows its equality key and the frequency of that key;
    split keys are not shown. (For each node, take the split key to be any key that separates
    the keys in the left and right subtrees.)
    The costs of $\hfigtree{huang small}{a}$ and $\hfigtree{huang small}{b}$ are $209$ and $210$, respectively,
    but in $\hfigtree{huang small}{a}$, the total weight of the keys is larger by 2.
  }\label{fig: huang small}
\end{figure}


Figure~\ref{fig: huang small} shows two possible subtrees $\hfigtree{huang small}{a}$ and 
$\hfigtree{huang small}{b}$ for $(I_9, 2)$, each with seven nodes. 
By calculation, subtree $\hfigtree{huang small}{b}$ costs 1 more than subtree $\hfigtree{huang small}{a}$ for $(I_9, 2)$.
(Indeed, key $\keyCzero$ contributes $5$ units more to $\hfigtree{huang small}{b}$ than to $\hfigtree{huang small}{a}$,
while key $\keyBfour$ contributes $4$ units less to $\hfigtree{huang small}{b}$
than key $\keyDone$ contributes to $\hfigtree{huang small}{a}$.)

\begin{figure}\centering
   \noindent%
  \input{TREE_error_bad.tex}%

  \caption{Trees $\hfigtree{huang flaw}{a}$ and $\hfigtree{huang flaw}{b}$ for an instance of \gbsplit with 31-key interval $I_{31}$.
    Key order is lexicographic:
    $\keyAzero < \keyAone < \keyAtwo < \keyAthree < \keyBzero  < \cdots$. 
    As in Fig.~\ref{fig: huang small}, split keys are not shown.
    Huang and Wong's algorithm gives
    a tree of cost $1763$, such as $\hfigtree{huang flaw}{a}$, 
    but tree $\hfigtree{huang flaw}{b}$ costs $1762$. 
  }\label{fig: huang flaw}
\end{figure}


Although $\hfigtree{huang small}{b}$ costs 1 more than $\hfigtree{huang small}{a}$,
choosing subtree $\hfigtree{huang small}{b}$ instead of $\hfigtree{huang small}{a}$ can decrease the cost of the overall tree!
To see why, suppose that $\hfigtree{huang small}{a}$ occurs as a subtree of some tree $T^*$,
in which $\hfigtree{huang small}{a}$ has parent $\keyAthree$ and grandparent $\keyBfour$ as shown in the figure.
(See also Fig.~\ref{fig: huang flaw}.)
Consider replacing $\hfigtree{huang small}{a}$ and its two hole keys $\keyAthree$ and $\keyBfour$
by $\hfigtree{huang small}{b}$ and its  two hole keys $\keyAthree$ and $\keyDone$. 
This replacement decreases the cost of the entire tree by 1 unit, because
the contribution of $\keyCzero$ increases by $5$, swapping
$\keyBfour$ and  $\keyDone$ decreases the cost by $6$, and the contributions of
other nodes do not change. 
But a different calculation gives better intuition why Huang and Wong's algorithm fails.
The contribution of the subtree $\hfigtree{huang small}{a}$
to the overall cost equals the cost of $\hfigtree{huang small}{a}$
in isolation plus twice the weight of keys in $\hfigtree{huang small}{a}$
(because $\hfigtree{huang small}{a}$ has two ancestors).
The modification increases the cost of the subtree by 1 (so it is no longer optimal for its subproblem)
but decreases the total weight of its keys by 2.
Thus, the subtree’s contribution to the overall cost changes by $+1 - 2\cdot 2 = -3$.
This decrease of $3$ is more than the increase of 2 that comes from changing the key
$\keyBfour$ at the overall root to $\keyDone$, which is 2 units heavier.

Next we use this HW-subproblem to obtain the complete instance $(\Keys,p)$
for Theorem~\ref{thm: no optimal substructure}.
The instance has a 31-key interval $I_{31}$, which extends the previously considered interval $I_9$ 
by appending two ``neutral'' subintervals, with 7 and 15 keys.
Figure~\ref{fig: huang flaw} shows two trees $\hfigtree{huang flaw}{a}$ and $\hfigtree{huang flaw}{b}$ for $(\Keys,p)$. 
As shown there, the new keys are given weights
so that each of the two added subintervals (without any holes) has a self-contained, optimal balanced subtree.
To finish proving Theorem~\ref{thm: no optimal substructure},
  we prove that $(\Keys, p)$ has the necessary properties:      



\begin{lemma}\label{lemma: K p}
  Let $T^*$ be any optimal tree for this \gbsplit instance $(\Keys,p)$.
    At some node $N$ of $T^*$ the HW-subproblem $(I_9, 2)$ arises,
    but the subtree $T^*_N$ rooted at $N$ has cost at least 210
    for $(I_9, 2)$, while $\optstar{I_9, 2} \le 209$.
\end{lemma}


  To bound tree costs, define a \emph{key placement} (for a tree $T$) to be an assignment
  of the equality-test keys in $T$ to distinct nodes in the infinite rooted binary tree $T_\infty$.
  Define the \emph{cost} of the placement to be the average weighted depth of the placed keys,
  weighted according to the key weight-vector $p$.
  Each correct \gbsplit tree $T$ yields a placement of equal cost
  by placing each equality-test key in the same place in $T_\infty$ that it occupies in $T$.
  The converse does not hold, partly because placements can ignore the ordering of keys.

  By an exchange argument, a placement has minimum cost if and only if it puts
  the weight-22 key $\keyDone$ at depth 0,
  the fourteen weight-20 keys at depths 1--3,
  and the sixteen remaining (weight-10 and weight-5) keys at depth 4.
  By calculation, such a placement costs 1757.
  No placement costs less, so no tree costs less.
  Tree $\hfigtree{huang flaw}{b}$ almost achieves a minimum-cost placement---it fails only in that it places the weight-5 key at depth 5,
  so costs 1762, just 5 units more than the minimum placement cost.  

  \begin{Claim}\label{claim: placement}
    $T^*$ has the following structure: 
    \begin{itemize}
    \item[(i)] It places the fifteen keys of weight 20 or more at depths 0--3.
    \item[(ii)] It places the fifteen weight-10 keys at depth 4.
    \end{itemize}
  \end{Claim}

  Next we prove the claim.
  Since $T^*$ is optimal it costs at most 1762 (the cost of $\hfigtree{huang flaw}{b}$),
  so its placement also costs at most 1762.
  Suppose for contradiction that (i) doesn't hold.
  Then $T^*$ places a key $k$ of weight 20 or more at depth at least 4.
  Also, in depths 1--3, it either places at least one key $k'$ of weight 10,
  or places fewer than fifteen keys.
  In either case, by exchanging $k$ and $k'$, or just re-placing $k$ in depth 1--3,
  we can obtain a key placement that costs at least 10 units less than 1762.
  But this is impossible, as the minimum placement cost is 1757.
  So (i) holds.
  Now suppose for contradiction that (ii) doesn't hold.
  Then there is a weight-10 key $k'$ at depth 5 or more,
  and at most fifteen keys at depth 4, so $k'$ can be re-placed in depth 4,
  yielding a key placement that costs 10 less, which is impossible.
  This proves the claim.

  Key placements ignore the ordering of keys.
  The following \emph{order property} captures the restrictions on key placements due to the ordering.

	  \begin{quote}
      \emph{
        Let $T$ be any correct \gbsplit tree. 
        Let {$P$} and {$P'$} be nodes in $T$ with equality-test keys $k$ and $k'$.
        Let {$Q$} be the least-common ancestor of {$P$} and {$P'$}.
        If {$P$} is in {$Q$}'s left subtree, and {$P'$} is 
  	  in {$Q$}'s right subtree, then $k<k'$.}
	  \end{quote}

  The property holds simply because $k$ and $k'$ are separated by $M$'s split key.
  
  Fix any optimal tree $T^*$ for $(\Keys,p)$. 
  Claim~\ref{claim: placement} imposes stringent constraints on the depth of all keys in $T^*$,
  except for the weight-5 key $\keyCzero$.  There are two cases:
  \begin{outercases}
  \item[1]\emph{$T^*$ places $\keyCzero$ at depth 4.}
    With Claim~\ref{claim: placement}, this implies that
    $T^*$ is a complete balanced binary tree of depth 4 (like $\hfigtree{huang flaw}{a}$),
    whose sixteen depth-4 nodes hold the fifteen weight-10 keys and $\keyCzero$.
    By the order property, these depth-4 keys are ordered left to right,
    just as they are in $\hfigtree{huang flaw}{a}$, with the left-most four nodes at depth 4 having keys
    $\keyBzero$, $\keyCzero$, $\keyDzero$, and $\keyEzero$.
    
    The left spine has only five nodes.  By the order property, all five keys less than $\keyCzero$
   cannot be elsewhere than on the spine.  So \emph{$\keyDone$ is not on the left spine.}
    
    Let $M$ be the parent of sibling leaves $\keyDzero$ and $\keyEzero$. Since $\keyDzero < \keyDone < \keyEzero$,
    by the order property, $\keyDone$ must lie on the path from $M$ to the root.
    Since $\keyDone$ is not on the left spine,
    and $M$ is the only node on this path that is not on the left spine, $\keyDone$ must be $M$.
    So $\keyDone$ has depth 3 in $T^*$. Now exchanging $\keyDone$ with the root key gives a placement
    that costs at least 6 less, that is, at most $1762-6 < 1757$,
    which is impossible as the minimum placement cost is 1757. So Case~1 cannot happen.

  \item[2]\emph{$T^*$ places $\keyCzero$ at depth 5.}
    Let $L_0, L_1, \ldots, L_\ell$ be the left spine of $T^*$, starting at the root.
    Take $T'$ to be the subtree of $T^*$ rooted at $L_2$.
    By Claim~\ref{claim: placement}, $T^*$ has fifteen depth-4 nodes, holding the fifteen weight-10 keys.
    By the order property, these depth-4 keys are ordered left to right within their level and
	at most twelve of them are not in $T'$.
    This implies that the weight-10 keys $\keyBzero$, $\keyD0$ and $\keyEzero$ must be in $T'$.
		  
    The next larger weight-10 key, $\keyNzero$, cannot be in $T'$.
    Indeed, if it were, then by the order property, all keys less than or equal to $\keyNzero$
    would be in $T'\cup \{L_1, L_0\}$.
    But there are twelve keys less than or equal to $\keyNzero$ and at most eight keys in $T'$.
	
    We now focus on the cost of $T'$.
	By the previous two paragraphs, $T'$ has exactly three keys at depth $2$, namely $\keyBzero$,  $\keyDzero$, and $\keyEzero$.
    By the order property and the assumption for Case 2,
    $\keyCzero$ must be (the only key) at depth 3 in $T'$ (as the child of either $\keyBzero$ or $\keyDzero$). 
    By Lemma~\ref{lemma: depth}, the three keys at depths 0 and 1 in $T'$ have weight 20 or 22. 
	Therefore,
    by calculation, \emph{the cost of $T'$ is at least 210 (see Fig.~\ref{fig: huang small}).}

    Since $\keyEzero$ is in $T'$, by the order property, all eight keys less than $\keyEzero$ are in $T' \cup\{L_0, L_1\}$.
    That is, $T'\cup\{L_0, L_1\}$ contains at least the 9 keys in $I_{9}$.
    But (as observed above) $T'$ has seven nodes. So $T'\cup\{L_0, L_1\}$ contains exactly the 9 keys in $I_{9}$,
    and the HW-subproblem solved by $T'$ must be $(I_9, 2)$.
    As observed above, $T'$ costs at least 210.
    But tree $\hfigtree{huang small}{a}$  (Fig.~\ref{fig: huang small}) of cost 209 also solves $(I_9, 2)$,
   so $\optstar{I_9, 2} \le 209$.
  \end{outercases}
  This proves Lemma~\ref{lemma: K p} and Theorem~\ref{thm: no optimal substructure}. \hfill\qed
  \vspace{-1.1em}
  
  ~
\end{proof}


  We prove one final utility lemma before we prove Theorem~\ref{thm: huang flaw}.
  Consider any execution of Huang and Wong's algorithm on the input $(\Keys, p)$
  defined in the proof of Theorem~\ref{thm: no optimal substructure}.
  Let $T = \hwTree(I_{31}, 0)$ be the algorithm's solution.  


  \begin{lemma}\label{lemma: I9}
    If $T$ contains a node $N$ whose HW-subproblem is $(I_9, 2)$,
    then the subtree $\hwTree(I_9, 2)$ rooted at $N$ 
    costs at most 209 for $(I_9, 2)$.
  \end{lemma}      


  \begin{proof}
    Abusing notation, for $1\le i < j \le 9$,
    let $[i,j]$ denote the $i$th through $j$th keys in interval
    $I_9$, as shown in Fig.~\ref{fig: huang flaw}.
    (See also Fig.~\ref{fig: huang small} for intuition.)

    \begin{table}[t]\centering
    \renewcommand*{\arraystretch}{1.2}
    \[
      \begin{array}
        {@{~} c @{~~~} || @{~~} c @{~~} | @{~~} c @{~~~~} c  @{~~} | @{~~} c @{~~~~} c}
        I, h & s,h_1, h_2 & \text{left} & \text{right}
      & \text{cost for $\hwTree(I,h)$} & \text{holes}
      \\ \hline
      {[1,5]}, 4 & \multicolumn{2}{c}{\text{}}  &
      & 10 & \keyAone, \keyAtwo, \keyAthree, \keyBfour
      \\
      {[6,6]}, 0 & \multicolumn{3}{c| @{~~}}{\text{singleton cases}}  
      & 5 & \text{none}
      \\
      {[7,8]}, 1 & \multicolumn{2}{c}{\text{}}  &
      & 10 & \keyDone
      \\
      {[9,9]}, 0 & \multicolumn{2}{c}{\text{}} &
      & 10 & \text{none}
      \\ \hline
      {[1,6]}, 3 & \keyCzero, 4, 0 & {[1,5]}, 4 & {[6, 6]}, 0
      & 10 + 5 + 35 = 50 & \text{three of } \keyAone, \keyAtwo, \keyAthree, \keyBfour
      \\
      {[7, 9]}, 0 & \keyEzero, 1, 0 & {[7,8]}, 1 & {[9, 9]}, 0
      & 10 + 10 + 42 = 62 & \text{none}
      \\ \hline
      {[1,9]}, 2 & \keyDzero, 3, 0 & {[1,6]}, 3 & {[7,9]}, 0
      & \le 50 + 62 + 97 = 209 &\text{two of } \keyAone, \keyAtwo, \keyAthree, \keyBfour
      \\ \hline
    \end{array}
    \]
    \caption{The HW-subproblems used to solve HW-subproblem $(I_9, 2)$, with $I_9=[1, 9]$.}\label{table}
    \end{table}
      
    Consider Table~\ref{table}. 
    Each row of the table is for one HW-subproblem $(I, h)$ (shown in the leftmost column),
    and demonstrates that the cost of the tree $\hwTree(I,h)$ computed by the algorithm
    for that subproblem is as shown in the fifth column (``cost for $\hwTree(I,h)$'').
    The last column lists the keys that are holes in $\hwTree(I, h)$.
    The first four rows are singleton cases (key sets of size one), and
	their correctness and optimality can be verified by straightforward inspection.
    For each subsequent row, the second column gives
    one of the triples $(s, h_1, h_2)$ considered by the algorithm
    for the given HW-subproblem $(I,h)$, where $s$ is the split key,
    and $h_1$ and $h_2$ are the numbers of holes allocated to the left and right subtrees.
    Columns ``left'' and ``right'' show the left and right
    HW-subproblems that follow from that choice of $(s, h_1, h_2)$,
    and column ``cost for $\hwTree(I,h)$'' gives the cost of tree $T(\KLT, \hLT, \hGE)$
	resulting from that choice. 
    Likewise the final column ``holes'' describes the possible hole sets
    (in order to achieve the given cost,
    covering all ways to break ties).
    For the HW-subproblems in rows five and six, the choices of $(s, h_1, h_2)$ in the table are optimal.
    For the seventh subproblem, the cost of 209 is an upper bound (in fact it is optimal, but we don't need that here).
    Each row can be verified by manual computation
    assuming inductively that the previous rows are correct.
	
    
    \begin{linenomath*}
    To illustrate how to verify the rows, we explain the
    information included in the 5th row, for HW-subproblem $(I,h) = ([1,6],3)$.
    This subproblem involves interval $[1,6]$ that consists of
    keys $\keyAone$, $\keyAtwo$, $\keyAthree$, $\keyBzero$, $\keyBfour$, $\keyCzero$,
    with $3$ of the keys being holes.
    For the choice $(s,h_1,h_2) = (\keyCzero,4,0)$ in the algorithm (the 2nd column),
    the left and right HW-subproblems will be $([1,5],4)$ and $([6,6],0)$ (the 3rd and 4th column).
    Their solutions are  summarized in the 1st and 2nd row of the table.
    (These solutions are: $\hwTree([1,5],4)$ contains only node $\keyBzero$,
    and $\hwTree([6,6],0)$ contains only node $\keyCzero$.)
    The algorithm will then choose any key from $\keyAone$, $\keyAtwo$, $\keyAthree$, $\keyBzero$,  
    as the equality key in the root of tree $T(\keyCzero,4,0)$,
    since they all have the same weight $20$.
    The weight of $T(\keyCzero,4,0)$ is then $35$, so its cost will be $50$,
    and the holes will be any three keys among $\keyAone$, $\keyAtwo$, $\keyAthree$, $\keyBzero$.
    (Note: another choice in the algorithm that gives the same tree 
    is $(s,h_1,h_2) = (\keyBfour,3,1)$.)
    As claimed in the paragraph above, this tree $T(\keyCzero,4,0)$ is an optimal solution for
    HW-subproblem $([1,6],3)$, that is $T(\keyCzero,4,0) = \hwTree([1,6],3)$.
    Indeed, $T(\keyCzero,4,0)$ is the only tree for $([1,6],3)$ that contains only one key of weight $20$,
    and any tree that has two keys of weight $20$
    will have cost at least $60$.
    \hfill\qed
    \vspace*{-1em}

    ~
  \end{linenomath*}
\end{proof}    


By Lemmas~\ref{lemma: K p} and~\ref{lemma: I9}, the tree $T$ computed by Huang and Wang's
algorithm for $(\Keys, p)$ cannot be optimal: Lemma~\ref{lemma: K p} states that
\emph{all} optimal trees for $(\Keys, p)$ contain a node with a certain property, while
Lemma~\ref{lemma: I9} states that $T$ does not contain such a node.
This proves Theorem~\ref{thm: huang flaw}.


For empirical verification,
note that executing the algorithm on $(\Keys, p)$, via the Python code in Appendix~\ref{sec: code huang},\footnote
{
  The code there is modified to return, for each subproblem,
  not just one tree but all ``candidate'' trees of minimum cost,
  where a candidate is a tree that the recurrence could consider by \emph{any} way of breaking ties.
  If any way of breaking ties will solve all of the relevant subproblems optimally,
  this simulation will find it.}
returns a tree of cost 1763.
This tree is not optimal, as $\hfigtree{huang flaw}{b}$ costs 1762.
The tree does have the HW-subproblem $(I_9,2)$,
and executing the algorithm directly on that subproblem does return a tree of cost 209.


\paragraph{Remark on Chen and Liu's algorithm for \myproblem{multiway}\!\! \gbsplit.}
Chen and Liu's algorithm~\cite{chen_optimal_1991} and analysis
are patterned directly on Huang and Wong's,
and the proofs they present also conflate (their equivalents of)  $\optstar{I,h}$ and $\opt{I,H}$,
leading to the same problems with optimal substructure.
For example, Property 1 of~\cite{chen_optimal_1991} states 
\emph{``Any subtree of an optimal $(m+1)$-way generalized split tree is optimal.''}
They do not define ``optimal'',
so their Property 1 has the same problem as Huang and Wong's Lemma 1:
it is true if ``optimal'' means ``with respect to their equivalent of $\opt{I,H}$'',
but does not necessarily hold if ``optimal'' means ``with respect to their equivalent of $\optstar{I,h}$''.
Lemmas 2, 3 and 4 of~\cite{chen_optimal_1991},
which state\\
the recurrence relations for their dynamic program,
are direct generalizations of Huang and Wong's Lemma 4.
Their recurrence chooses equality keys by first finding optimal subtrees for the children,
then taking the equality keys to be the least-likely keys that are not equality keys in the children's subtrees.
As pointed out in the proof of Theorem~\ref{thm: huang flaw},
correctness of this approach requires the optimal-substructure property
to hold with respect to $\optstar{I,h}$.  But it does not.
For these reasons, their proof of correctness is not valid.
We believe that their algorithm for \myproblem{multiway} \gbsplit is also incorrect,
but describing their algorithm and analysis in detail,
and giving a complete counter-example, are out of the scope of this paper.

\smallskip


%% file: TREE_error_d=1.tex
{
\scriptsize
 
\begin{forest}
  for descendants={gbst,},,,
  [, phantom, s=1em
  [{\normalsize ~}, nodraw [, phantom [, phantom [, phantom [, phantom][, phantom]]]]]
   [\keyBfour\\20, bright, s sep=5em, 
    [\keyAthree\\20, s sep=5em
      [\keyAtwo\\20, name=A2, tikz={
        \begin{scope}[on background layer]
          \node at (A2.west)[xshift=-2.5em,yshift=3em] {{\large $\hfigtree{huang small}{a}$}};
          \node
          [draw, dashed, fill=gray, fill opacity=.1, 
              fit=()(!11)(!ll), 
              semicircle, rounded corners=1.5em, 
              inner sep=0em,
              yshift=-1.4ex,
              xscale=0.6,
              yscale=0.98
          ]
          {};
        \end{scope}
      }
          [\keyAone\\20, 
              [\keyBzero\\10, ]
              [\keyCzero\\ 5, bright]
          ]
          [\keyDone\\22, bright
              [\keyDzero\\10,]
              [\keyEzero\\10,]
          ]
      ]
      [, opacity=0, edge=dotted]
  ]
  [, opacity=0, edge=dotted]
   ]
[, phantom [, phantom [, phantom [, phantom [, phantom][, phantom]]]]]
[, phantom [, phantom [, phantom [, phantom [, phantom][, phantom]]]]]
[{\normalsize ~}, nodraw [, phantom [, phantom [, phantom [, phantom][, phantom]]]]]
[\keyDone\\22, bright, s sep=5em, 
    [\keyAthree\\20, s sep=5em
      [\keyAtwo\\20, name=XX, tikz={
        \node at (XX.west)[xshift=-2.5em,yshift=3em] {{\large $\hfigtree{huang small}{b}$}}; 
      \begin{scope}[on background layer]
        \node 
          [draw, dashed, fill=gray, fill opacity=.1, 
              fit=()(!11)(!ll), 
              semicircle, rounded corners=1.5em, 
              inner sep=0em,
              yshift=-1.4ex,
              xscale=0.6,
              yscale=0.98
          ]
        {};
      \end{scope}
      }
        [\keyAone\\20, 
          [\keyBzero\\10, ]
          [\keyCzero\\5, bright, phantom]
        ]
        [\keyBfour\\20, bright,
          [\keyDzero\\10,
            [\keyCzero\\ 5, bright]
            [X\\0, phantom]
          ]
          [\keyEzero\\10,]
        ]
      ]
      [, opacity=0, edge=dotted]
    ]
    [, opacity=0, edge=dotted]
  ]
]
\end{forest}
}


%% file: TREE_error_bad.tex
{
\scriptsize

\begin{forest}
  for tree={gbst,},,,
  [\keyBfour\\20, bright, name=root,
    [\keyAthree\\20, 
      [\keyAtwo\\20, tikz={
        \begin{scope}[on background layer]
          \node
          [draw, dashed, fill=gray, fill opacity=.05, 
              fit=()(!11)(!ll), 
              semicircle, rounded corners=1.5em, 
              inner sep=0em,
              yshift=-1.4ex,
              xscale=0.6,
              yscale=0.98
          ]
          {};
        \end{scope}
      }
          [\keyAone\\20, 
              [\keyBzero\\10, name=left]
              [\keyCzero\\ 5, bright]
          ]
          [\keyDone\\22, bright
              [\keyDzero\\10,]
              [\keyEzero\\10,]
          ]
      ]
      [\keyFtwo\\20, ,
          [\keyFone\\20
              [\keyNzero\\10,]
              [\keyPzero\\10,]
          ]
          [\keyQone\\20, 
              [\keyQzero\\10,]
              [\keyRzero\\10,]
          ]
      ]
  ]
  [\keyVthree\\20,
      [\keyTtwo\\20,
          [\keySone\\20,
              [\keySzero\\10,]
              [\keyTzero\\10,]
          ]
          [\keyUone\\20, 
              [\keyUzero\\10,]
              [\keyVzero\\10,]
          ]
      ]
      [\keyXtwo\\20, 
          [\keyWone\\20,
              [\keyWzero\\10,]
              [\keyXzero\\10,]
          ]
          [\keyYone\\20,
              [\keyYzero\\10,]
              [\keyZzero\\10,]
          ]
      ]
  ]
]
\node at (current bounding box.north)[xshift=-3.5em,yshift=-0.2em]{{\large $\hfigtree{huang flaw}{a}$}}; 
\end{forest}

\medskip

\begin{forest}
  for tree={gbst,},,
  [\keyDone\\22, bright, name=root
    [\keyAthree\\20, 
      [\keyAtwo\\20, tikz={
      \begin{scope}[on background layer]
        \node 
          [draw, dashed, fill=gray, fill opacity=.05, 
              fit=()(!11)(!ll), 
              semicircle, rounded corners=1.5em, 
              inner sep=0em,
              yshift=-1.4ex,
              xscale=0.6,
              yscale=0.98
          ]
        {};
      \end{scope}
      }
        [\keyAone\\20, 
          [\keyBzero\\10, name=left]
          [\keyCzero\\5, bright, phantom]
        ]
        [\keyBfour\\20, bright,
          [\keyDzero\\10,
            [\keyCzero\\ 5, bright]
            [X\\0, phantom]
          ]
          [\keyEzero\\10,]
        ]
      ]
      [\keyFtwo\\20, ,
        [\keyFone\\20
          [\keyNzero\\10,]
          [\keyPzero\\10,]
        ]
        [\keyQone\\20, 
          [\keyQzero\\10,]
          [\keyRzero\\10,]
        ]
      ]
    ]
    [\keyVthree\\20,
      [\keyTtwo\\20,
        [\keySone\\20,
          [\keySzero\\10,]
          [\keyTzero\\10,]
        ]
        [\keyUone\\20, 
          [\keyUzero\\10,]
          [\keyVzero\\10,]
        ]
      ]
      [\keyXtwo\\20, 
        [\keyWone\\20,
          [\keyWzero\\10,]
          [\keyXzero\\10,]
        ]
        [\keyYone\\20,
          [\keyYzero\\10,]
          [\keyZzero\\10,]
        ]
      ]
    ]
  ]
\node at (current bounding box.north)[xshift=-3.5em,yshift=-0.2em]{{\large$\hfigtree{huang flaw}{b}$}};
\end{forest}
}


%% file: 3_errors_spuler.tex


\newcommand{\EPS}{5}
\newcommand{\DELTA}{7}
\newcommand{\sfigtree}[2]{\ensuremath{T_{\ref{fig: spuler counter #1}#2}}}

\smallskip

This section concerns \twoWCST, the problem of computing an optimal \emph{two-way comparison} search tree,
given a set $\Keys$ of $n$ keys and their weight distribution $p$.
Such a tree $T$ is a rooted binary tree, where each non-leaf node $N$ has two children,
as well as a key $k_N\in\Keys$ and a binary comparison operator (equality or less-than).
Denote such a node by $\compnode{=}{k_N}$ or $\compnode{<}{k_N}$, depending on which
comparison operator is used. The tree $T$ has $n$ leaves, each labeled with a unique key in $\Keys$.


\begin{figure}[t]\centering
  \noindent%
  \input{TREE_Spuler_example.tex}%

  \caption{
      A two-way-comparison search tree with keys $\Keys=\{1,2,3,4,5,6\}$.
      Below each leaf is its weight. The cost of this tree is 
	  $0.6\cdot 1 + 0.1\cdot 3  + 0.1\cdot 3 +  0.05\cdot 4  +  0.05\cdot 4 + 0.1\cdot 3 = 1.9$.
  }\label{fig: spuler example}
\end{figure}


The search for a query $\query$ in $T$ starts at the root. If the root is a leaf, the search halts.
Otherwise, it compares $\query$ to the root's key using the root's comparison operator,
then recurses left if the comparison succeeds, and right otherwise.
For the tree to be correct,\footnote
{Note that, in contrast to \gbsplit trees, there are no comparisons at leaf nodes. 
  For simplicity, we discuss here only the successful-queries variant,
  in which only queries in $\Keys$ are allowed.}
the search for any query $\query\in\Keys$
must end at the leaf that is labeled with $\query$.
Given an instance $(\Keys, p)$, 
  the problem is to find a tree that minimizes the weighted average depth of the leaves
  (in the case that $p$ is a probability distribution, this is the expected number of comparisons
  in a search for a query $\query$ drawn randomly according to $p$).
  Figure~\ref{fig: spuler example} shows an example.

Spuler's thesis proposed various  algorithms for \twoWCST and for \gbsplit,
for both the successful-queries variant and the general 
variant~\cite{Spuler1994Thesis}.\footnote{We remark 
  that Spuler~\cite[Section 4.8]{Spuler1994Thesis} 
  pointed out, and claimed to fix, several flaws in the \emph{pseudo-code}
  that Huang and Wong gave for their \gbsplit algorithm.
  Those flaws are relatively minor
  and do not include the deeper errors discussed in Sect.~\ref{sec: huang flaw}.}
Here we discuss the (successful-queries) \twoWCST algorithm that Spuler presented
as a modification of Huang and Wong's \gbsplit algorithm
in Section 6.4.1 of his thesis~\cite[Section 6.4.1]{Spuler1994Thesis}.
That section starts with the following remark:
\begin{quote}\em
  ``The changes to the optimal generalized binary split tree algorithm of Huang 
and Wong~\cite{StephenHuang1984} to produce optimal generalized two-way comparison trees are quite straight forward.''
\end{quote}
(``Generalized two-way comparison trees'' in the thesis
are two-way comparison search trees as defined herein.)
The remainder of his Section 6.4.1 sketches the code
  for the algorithm.  His Appendix A.4.1.~gives complete code.
  Spuler does not explicitly define the dynamic program or recurrence that he has in mind,
  however, it is implicitly defined by his algorithm as described below.
In addition to lacking proofs of correctness,
these algorithms have not appeared in any peer-reviewed publication,
although Spuler did refer to them in his journal paper~\cite{Spuler1994Paper},
and they have been cited in the literature as the first polynomial-time algorithms for \twoWCST~\cite{Anderson2002}.

Following Huang and Wong, Spuler's algorithms are based on a dynamic program
where each subproblem is specified by an interval of keys and a number of holes,
and each subproblem is solved using a recurrence relation. In the remainder of the section, we prove that
the dynamic program is flawed:


\begin{theorem}\label{thm: spuler flaw}
There is an instance $(\Keys, p)$ of \twoWCST for which  the dynamic program used by
Spuler's \twoWCST algorithm~\cite[Section 6.4.1]{Spuler1994Thesis}  has the following flaws:
for some subproblems, the recurrence relation is incorrect and the algorithm computes non-optimal solutions.
\end{theorem}


Note that Theorem~\ref{thm: spuler flaw} does not imply that the algorithm is incorrect,
in the sense that it gives an incorrect solution to some full instance (where the number $h$ of holes is $0$). 

Following Huang and Wong, the dynamic program implicit in Spuler's algorithm has a subproblem $(I, h)$
for each query interval $I$ and number of holes $h$. 
In what follows we call any such subproblem $(I, h)$ an \emph{S-subproblem}. 
The definition of a correct tree for an S-subproblem is a natural extension of the definition for full instances:
a correct tree for $(I,h)$ must have exactly $|I| - h$ leaves, 
each labeled with a unique key from $I$;
however, all keys in $I$ can be used as inequality-comparison keys.
We use $\optstar{I,h}$ to denote the minimum cost of any tree for S-subproblem $(I,h)$. 
The underlying flaw is the same as in Huang and Wong's dynamic program---S-subproblems do not have optimal substructure.

Given any S-subproblem $(I, h)$, where $I=[i,j]$ and the subproblem size $|I|-h$ is more than one, Spuler's
algorithm computes a tree $\tau(I,h)$ for S-subproblem $(I, h)$
  by combining trees $\tau(I',h')$ that it has computed for smaller S-subproblems, as follows:


\begin{center}
\begin{minipage}{4.7in}
\hrule
\medskip
\renewcommand{\labelenumii}{\theenumii}
  \renewcommand{\theenumii}{\theenumi.\arabic{enumii}.}

\begin{enumerate}

\item Construct one \emph{candidate tree} with an equality test at the root, as follows: 
  
  \begin{enumerate}[label*=\arabic*.]
	  
	  \item 
	 Let $\KEQ$ be a least-likely key in $I$ that is not a leaf in $\tau(I, h+1)$.

       \item The candidate tree has root $\compnode = \KEQ$ and  right subtree  $\tau(I, h+1)$.
	\end{enumerate}


\item 
  For $\KLT\in [i+1, j]$ and $(\hLT, \hGE)$ 
  s.t.~$\hLT + \hGE = h$, $\KLT-i - \hLT \ge 1$ and $j-\KLT+1-\hGE \ge 1$:

  \begin{enumerate}[label*=\arabic*.]
  \item Make a candidate tree 
    with root $\compnode < \KLT$
    and subtrees $\tau([i, \KLT-1], \hLT)$, $\tau([\KLT, j], \hGE)$.
  \end{enumerate}
  
\item Among the candidate trees so constructed, let $\tau(I,h)$ be one of minimum cost.

\end{enumerate}
\hrule
\end{minipage}
\end{center}

\paragraph{Remarks.}
  The algorithm is not hard to implement. Appendix~\ref{sec: code spuler} gives Python code (42 lines).
    As noted earlier, Spuler does not explicitly define his dynamic program or recurrence relation
    for $\optstar{I,h}$, however, it is implicitly defined by his algorithm and his assumption that each tree
  $\tau(I,h)$ is optimal for S-subproblem $(I, h)$ (so has cost $\optstar{I, h}$).

  Ties may arise in choosing the minimizers in Lines 1.1 and 3,
  but Spuler does not discuss ties.  We'll show that his recurrence relation
  is incorrect no matter how ties are broken.

  Given an S-subproblem $(I, h)$, Spuler's algorithm constructs its tree $\tau(I, h)$
  out of trees $\tau(I', h')$ that it built for smaller S-subproblems.
  This only works if S-subproblems have optimal substructure.
  To complete the proof of Theorem~\ref{thm: spuler flaw}, we show that they do not:


\begin{theorem}\label{thm: no optimal substructure spuler}
  There exists a \twoWCST\ S-subproblem $(I,h)$ with the following property.
  In every optimal tree $T^*$ for $(I, h)$,  there is at least one node $N$ such that,
  for the S-subproblem $(I_N,h')$ arising at $N$,
  the subtree $T^*_N$ rooted at $N$ in $T^*$
  does \emph{not} have minimum cost, $\optstar{I_N,h'}$, for that S-subproblem.\footnote
  {Note that $h'$ is $h$ plus the number of equality tests
    on the path from the root of $T^*$ to $N$.  The algorithm does not
    determine which keys are used in those equality tests until after it solves  $(I_N, h')$.}
  \end{theorem}


\begin{proof}  
Before we describe the full S-subproblem $(I,h)$,
we describe one smaller S-subproblem $(I',h')$ for which using a minimum-cost tree $T'$
can be a bad choice globally.
It is $(I_{8}, 1)$, with one hole and interval $I_8$ having keys $\{1,2,\ldots,8\}$ whose weights are as follows:

\begin{quote}\centering
  \begin{tabular}{|r||c|c|c|c|c|c|c|c|c||c|}\hline
    \text{key}    & 1 & 2 & 3 & 4 &  5 &  6 & 7 & 8 \\ \hline
    \text{weight} & \DELTA & \EPS & 0 & \EPS  & 0 & \EPS & 0 & \EPS  \\ \hline
  \end{tabular}
\end{quote}
\smallskip 

\begin{figure}[t]\centering
  \noindent%
  \input{TREE_Spuler_counter_example_4.tex}%

  \caption{
    Three trees (circled and lightly shaded) for S-subproblem $(I_8, 1)$.
   \sfigtree 4 a has cost 49 and weight 22. 
    \sfigtree 4 b and \sfigtree 4 c have cost 50 but weight 20.
    \sfigtree 4 a is optimal for $(I_8, 1)$.
    Among trees that don't contain the weight-7 key 1,
    trees \sfigtree 4 b and \sfigtree 4 c have minimum cost.
    Subtrees marked with $0$ (dark shaded) contain keys of weight $0$.
  }\label{fig: spuler counter 4}
\end{figure}    


Figure~\ref{fig: spuler counter 4} shows three possible 
subtrees \sfigtree 4 a,  \sfigtree 4 b, and \sfigtree 4 c for $(I_8, 1)$.
By inspection, \sfigtree 4 a has cost 49 for S-subproblem $(I_{8}, 1)$,
while \sfigtree 4 b and \sfigtree 4 c cost 50 but weigh 2 units less.
Suppose, in a larger tree, that \sfigtree 4 a occurs as the left child of a node $N$,
as shown in Fig.~\ref{fig: spuler counter 4}. Let $T_N$ be the subtree rooted at $N$.
Suppose that the interval of the left child of $T_N$ (that is, the root of \sfigtree 4 a)
contains all keys $1,2,\ldots,8$. (For example, $N$ might be $\compnode{<}{9}$.)
Then replacing \sfigtree 4 a by \sfigtree 4 b would reduce the overall cost by at least 1 unit.
This is because the contribution of \sfigtree 4 a to the cost of $T_N$
is not the cost of  $\sfigtree 4 a$; rather, it is its cost \emph{plus its weight},
and the cost plus weight of \sfigtree 4 b is 1 unit less.

Next we construct the full S-subproblem $(I, h)=(I_{15},2)$
for Theorem~\ref{thm: no optimal substructure spuler}.
It has two holes, and extends the above S-subproblem $(I_{8}, 1)$
to a larger interval $I_{15} = \{1,2,\ldots,15\}$ with the following symmetric weights:

\begin{linenomath*}
\[\setlength\arraycolsep{4.5pt}
  \begin{array}{|r||c|c|c|c|c|c|c|c|c|c|c|c|c|c|c|c|c|c|}\hline
    \text{key}    & 1 & 2 & 3 & 4 &  5 &  6 & 7 & 8 & 9 & 10 & 11 & 12 & 13 & 14 & 15 
    \\ \hline
    \text{weight} & \DELTA & \EPS & 0 & \EPS & 0 & \EPS & 0 & \EPS & 0 & \EPS & 0 & \EPS & 0  & \EPS & \DELTA
    \\ \hline
  \end{array}
\]
\end{linenomath*}
We use the following terminology to distinguish the different types of keys in a subtree.
Given an S-subproblem $(I',h')$ of $(I,h)$, and a tree $T'$ for $(I', h')$, the
keys of $I'$ that appear in the leaves of $T'$ are \emph{$T'$-queries}.
The other keys in interval $I'$, which are holes in $T'$, are \emph{$T'$-holes}.
(We don't introduce new terminology for the comparison keys in $T'$.)
We drop the prefix $T'$ from these terms when it is understood from context.

To analyze $(I_{15},2)$ we need some utility lemmas. We start with
one that will help us characterize how weight-0 queries increase costs.
This lemma (Lemma~\ref{lemma: depth} below) is in fact general and it
holds for S-subproblems of an arbitrary instance of \twoWCST.
Define two integer sequences $\braced{d_m}$ and $\braced{e_m}$,
as follows: $d_1 = 0$, $d_2 = 3$, $e_1 = 0$, $e_2 = 2$, $e_3 = 6$, and

\begin{linenomath*}
\begin{align*}
	d_m &~=~ m\, + \,\min\,\{d_i + d_{m-i} : 1\le i < m \}\,\quad\text{for $m\ge 3$,}
	\\
	e_m &~=~ m\, + \,\min\,\{d_i + e_{m-i} : 1\le i < m \}\,\quad\text{for $m\ge 4$.}
\end{align*}
\end{linenomath*}
 By calculation,
$d_3 = 6$, $d_4 = 10$, $d_5 = 14$,
$e_4 = 9$, $e_5 = 13$, and $e_6 = 18$.

Consider a tree $T'$ for an S-subproblem of some 
arbitrary instance of \twoWCST (not necessarily our specific instance $(\Keys,p)$). 
A subset $Q$ of $T'$-queries will be called \emph{$T'$-separated} (or simply \emph{separated}, if $T'$ is
understood from context) if for any two $k,k'\in Q$, with $k < k'$,
there is a $T'$-query $k''$ that separates them, that is
$k < k'' < k'$. Also, if $Q\setminus\braced{f}$ is $T'$-separated for some $f\in Q$,
then we say that $Q$ is \emph{nearly $T'$-separated}.


\begin{lemma}\label{lemma: depth}
Let $T$ be a tree for an S-subproblem of some arbitrary instance of \twoWCST. 
Let $Q$ be a set of $T$-queries and $m=|Q|$.
\emph{(i)} 
	If $Q$ is $T$-separated then  
the total depth (i.e., the sum of the depths) in $T$ of the keys in $Q$ is at least $d_m$.
\emph{(ii)} 
	If $Q$ is nearly $T$-separated then the total depth in $T$ of the keys in $Q$ is at least $e_m$.
\end{lemma}

The proof of Lemma~\ref{lemma: depth} is a straightforward induction---we 
postpone it to the end of this section, and proceed with our analysis.


Now we focus our attention on our instance $(\Keys,p)$, and
we characterize the weights and costs of optimal subtrees for certain subproblems.
For $1\le \ell\le 14$,
let $I_\ell=\{1,2,\ldots,\ell\}$ denote the subinterval of $I_{15}$ containing its first $\ell$ keys.
These keys have $\ell$ weights (in order) $\{7, 5, 0, 5, \ldots\}$:
one key of weight 7, then $\lfloor \ell/2\rfloor$ even keys of weight 5,
separated by odd keys of weight 0.
Let $\nofposvalues{\ell} =  1+\lfloor \ell/2\rfloor$ be the number of positive-weight keys in $I_\ell$.
Note that each S-subproblem $(I_\ell, h')$
can be solved by a tree with $\nofposvalues{\ell}-h'$ positive-weight queries,
having $h'$ (positive-weight) hole keys.


\begin{lemma}\label{lemma: T 4}
  Consider any S-subproblem $(I_\ell, h')$ with $\ell \le 14$ and $\nofposvalues{\ell}-h' = 4$.
  Let $T'$ be an optimal tree for $(I_\ell, h')$.
  Then $T'$ has weight 22 and cost 49 (like \sfigtree 4 a).
\end{lemma}

\begin{proof}
As $T'$ is fixed throughout the proof, the terms \emph{holes}, \emph{queries}, and \emph{separated}, 
mean \emph{$T'$-holes}, \emph{$T'$-queries}, and \emph{$T'$-separated}
as defined earlier, unless otherwise specified. 

Let $h_0$ be the number of weight-0 holes and $\nofposvalues{q}$ the number of queries with positive weight.
We use the following facts about $T'$.
    \begin{description}  
	 \setlength{\itemsep}{0.75em}
	
      \item[(F1)] \emph{$T'$ costs at most 49.}
        Indeed, one way to solve $(I_\ell,h')$ is as follows:
        take the $h'$ rightmost weight-5 keys in $I_\ell$ to be the holes,
        then handle the remaining $\nofposvalues{\ell}-h'=4$ queries with positive weight
        (queries $1,2,4,6$), along with any weight-0 queries $3,5,7,\ldots$,
        using tree \sfigtree 4 a, at cost 49. 
        
    \item[(F2)] \emph{$\nofposvalues{q} = 4+h_0$}. 
		This follows by simple calculation:
		$\nofposvalues{q} = \nofposvalues{\ell}-(h'-h_0) = 4+h_0$. 

    \item[(F3)]
      \emph{$T'$ does not contain four separated weight-5 queries.}
      Indeed, otherwise, by Lemma~\ref{lemma: depth}, $T'$ would cost at least $5\cdot d_4 = 50 > 49$,
      contradicting (F1).
	  
    \end{description}
    
    To finish we show that $T'$ costs at least 49. Along the way we show it has weight 22.
    
    \begin{outercases}
    \item[1] First consider the case that $h_0 = 0$.
      By (F2), \emph{there are $4$ positive-weight queries in $T'$.}
      Since $h_0=0$, all weight-0 keys are queries in $T'$,
      so the set of all weight-5 queries in $T'$ is separated,
      and by (F3), \emph{there are at most three such queries}.
      The fourth positive-weight query must be the weight-7 query, query 1.
      So  \emph{the positive-weight queries in $T'$ are the weight-7 query and three separated weight-5 queries.}

      So $T'$ has total weight 22, as desired.
      Further, by Lemma~\ref{lemma: depth}, the four positive-weight queries in $T'$ 
      have total depth at least $e_4$ in $T'$.
      So $T'$ costs at least $5\cdot e_4 + (7-5)\cdot j =  45+2j$,
      where $j$ is the depth of the weight-7 query.
      If $j\ge 2$, by the previous bound, $T'$ costs at least 49, and we are done.
      In the remaining case we have $j = 1$ (as $j=0$ is impossible),
      so the weight-7 query is a child of the root.
      The three weight-5 queries are in the other child's subtree (and are a separated subset there),
      so by Lemma~\ref{lemma: depth} have total depth at least $d_3=6$
      in that subtree, and therefore total depth at least 9 in $T'$.
      So the total cost of $T'$ is at least $7 + 5\cdot 9 > 49$, contradicting (F1). 

    \item[2] In the remaining case $h_0 \ge 1$.
      By (F2), there are $\nofposvalues{q} = 4+ h_0$ positive-weight queries in $T'$.
      Let $q_5 \ge \nofposvalues{q}-1$ be the number of weight-5 queries in $T'$.
      Since all but $h_0$ of the weight-0 queries are in $T'$,
      there is a separated set of $q_5-h_0$ weight-5 queries in $T'$.  By (F3),
      $q_5 - h_0 \le 3$.

      \begin{linenomath*}
      This (with $\nofposvalues{q} = 4+h_0$ and $q_5\ge \nofposvalues{q}-1$) implies 
	$q_5 = h_0 + 3 = \nofposvalues{q} - 1$. This implies that the weight-7 query is in $T'$,
      along with some $q_5-h_0 = 3$ separated weight-5 queries.
      Reasoning as in Case~1, the cost of these four queries alone is at least 49.
      But $T'$ contains at least one additional weight-5 query
      (as $q_5 = 3+h_0 > 3$), so $T'$ costs strictly more than 49,
      contradicting (F1). Thus Case~2 cannot actually occur.
      \hfill\qed\vspace{-1em}

      ~
    \end{linenomath*}
    \end{outercases}
\end{proof}


\begin{figure}[t]\centering
  \noindent%
  \input{TREE_Spuler_counter_example_5.tex}%

  \caption{
    Trees $\sfigtree 5 a$ and  $\sfigtree 5 b$, with five positive weight queries.
    Tree $\sfigtree 5 a$ has cost 69 and weight 27.
    Tree  $\sfigtree 5 b$ has cost 70 and weight 25.
  }\label{fig: spuler counter 5}
\end{figure}


\begin{lemma}\label{lemma: T 5}
  Consider any S-subproblem $(I_\ell, h')$ with $\ell \le 14$ and  $\nofposvalues{\ell}-h' = 5$.
  Let $T'$ be an optimal tree for $(I_\ell, h')$.
  Then $T'$ has weight 27 and cost 69 (like \sfigtree 5 a in Fig.~\ref{fig: spuler counter 5}).
\end{lemma}

\begin{proof}
Again, throughout the proof, unless otherwise specified,
the terms \emph{holes}, \emph{queries}, and \emph{separated},
are all with respect to $T'$. Let $h_0$ be the number of weight-0 holes and 
$\nofposvalues{q}$ the number of queries with positive weight.
We use the following facts about $T'$.

\begin{description}
	
\item[(F4)] \emph{$T'$ costs at most 69.}
        Indeed, one can solve $(I_\ell,h')$ is as follows:
        take the $h'$ rightmost weight-5 keys in $I_\ell$ to be the holes,
        then handle the remaining $\nofposvalues{\ell}-h'=5$ queries with positive weight
        (queries $1,2,4,6,8$), along with any weight-0 queries $3,5,7,\ldots$,
        using tree \sfigtree 5 a at cost 69.  

\item[(F5)] \emph{$\nofposvalues{q} = 5+h_0$}. 
	This follows by straightforward calculation:
	$\nofposvalues{q} = \nofposvalues{\ell}-(h'-h_0) = 5+h_0$.
	
\item[(F6)]  \emph{$T'$ does not contain five separated weight-5 queries.} 
      Indeed, otherwise, by Lemma~\ref{lemma: depth},
      $T'$ would cost at least $5\cdot d_5 = 70 > 69$, a contradiction.
	  
\end{description}
    
    To finish, we show that $T'$ has cost at least 69.  Along the way we show it has weight 27.

    \begin{outercases}
    \item[1] First consider the case that $h_0 = 0$.
      By (F5), \emph{there are $5$ positive-weight queries in $T'$.}
      Also, since $h_0=0$, all weight-0 keys are queries in $T'$,
      so the set of all weight-5 queries in $T'$ is separated,
      and by (F6), \emph{there are at most four of them}.
      The fifth positive-weight query must be the weight-7 query, query 1.
      So \emph{the positive-weight queries in $T'$ are the weight-7 query and four separated weight-5 queries.}

      So $T'$ has total weight 27. Further, by Lemma~\ref{lemma: depth}, the five positive-weight queries in $T'$ 
      have total depth at least $e_5$ in $T'$. So $T'$ costs at least $5\cdot e_5 + (7-5)\cdot j =  65+2j$,
      where $j$ is the depth of the weight-7 query.
      If $j\ge 2$ then, by the previous bound, $T'$ costs at least 69, and we are done.
      In the remaining case we have $j = 1$ (as $j=0$ is impossible),
      so the weight-7 query is a child of the root. The four weight-5 queries are in the other child's subtree 
	(and form a separated set there), so by Lemma~\ref{lemma: depth} have total depth at least $d_4=10$
      in that subtree, and therefore total depth at least 14 in $T'$.
      So the total cost of $T'$ is at least $7 + 5\cdot 14 > 69$, contradicting (F4). 

    \item[2] In the remaining case, $h_0 \ge 1$.
      By (F5), there are $\nofposvalues{q} = 5+ h_0$ positive-weight queries in $T'$.
      Let $q_5 \ge \nofposvalues{q}-1$ be the number of weight-5 queries in $T'$.
      Since all but $h_0$ of the weight-0 queries are in $T'$,
      there is a separated set of $q_5-h_0$ weight-5 queries in $T'$. By
      (F6), $q_5 - h_0 \le 4$.

      \begin{linenomath*}
      This (with $\nofposvalues{q} = 5+h_0$ and $q_5\ge \nofposvalues{q}-1$) 
      implies $q_5 = h_0 + 4 = \nofposvalues{q} - 1$.
      This implies that the weight-7 query is in $T'$,
      along with some separated set of $q_5-h_0 = 4$ weight-5 queries.
      Reasoning as in Case~1, the cost of these five queries alone is at least 69.
      But $T'$ contains at least one additional weight-5 query
      (as $q_5 = 4+h_0 > 4$), so $T'$ costs strictly more than 69,
      contradicting (F4). Thus Case~2 cannot actually occur.
      \hfill\qed 

      ~
    \end{linenomath*}
  \end{outercases}
\end{proof}


To conclude the proof of Theorem~\ref{thm: no optimal substructure spuler},
we prove that $(I_{15}, 2)$ has the necessary properties:\hspace*{-1em}
\begin{lemma}\label{lemma: obstacle}
  Let $T^*$ be any optimal tree for S-subproblem $(I_{15}, 2)$.
    Then $T^*$ has at least one node $N$ such that,
    for the S-subproblem $(I_N,|H_N|)$ arising at $N$,
    the subtree $T^*_N$ rooted at $N$ in $T^*$ does \emph{not}
    have minimum cost, $\optstar{I_N,|H_N|}$, for that S-subproblem.
\end{lemma}


  \begin{figure}[t]\centering
    \noindent%
  \input{TREE_Spuler_counter_example_full.tex}%

    \caption{
      Spuler's algorithm fails on the S-subproblem $(I_{15}, 2)$.
      The algorithm computes a tree of cost 116, such as \sfigtree F a above,
      but there are trees, such as \sfigtree F b, of cost 115.
      The two trees' left subtrees are \sfigtree 4 a and \sfigtree 4 c.
    }\label{fig: spuler counter F}
  \end{figure}
  
  
%
Throughout the proof, unless otherwise specified, the terms \emph{holes}, \emph{queries}, and \emph{separated},
are all with respect to $T^*$. We use the following properties of $T^*$:
  \begin{description}
	\setlength{\itemsep}{0.15em}
	  
  \item[(P1)]  \emph{$T^*$ costs at most 115.}
    Indeed, one way to solve $(I_{15}, 2)$ is to take the two weight-7 keys as holes,
    then use tree $\sfigtree F b$ in Fig.~\ref{fig: spuler counter F}, of cost 115.
    As $T^*$ is optimal, it costs at most 115. 

  \item[(P2)] \emph{The root of $T^*$ does a less-than comparison.}
    Indeed, by~\cite[Theorem~5]{Anderson2002}, since $T^*$ is optimal for its queries,
    if $T^*$ does an equality-test at the root, then the total query weight in $T^*$
    is at most four times the maximum query weight.
    But the total query weight in $T^*$ is at least $7\cdot 5 = 35$,
    while the maximum query weight is at most 7. 
    
  \item[(P3)] \emph{In $T^*$ there are seven positive-weight queries,
      and the set of weight-5 queries is separated (by weight-0 queries).}
    To show this, we show that no weight-0 key is a hole.
	Suppose otherwise for contradiction. Let $k'$ be a weight-0 hole.  
	We can assume without loss of generality that  $k'$ is not used
        in any node of $T^*$ as an inequality key, for otherwise we can
	modify $T^*$ to not use it, without changing its cost, by replacing it with the
	weight-5 key $k'' = k'+1$ (which could be a hole or a query).
	Since $k'$ is a $T^*$-hole, by definition, $k'$ also cannot be used as an equality key.
	So we can assume that $k'$ does not appear as a comparison key in $T^*$.
 	Let $k\in\{k'\pm 1\}$ be a weight-5 query in $T^*$.
    (Query $k$ exists in $T^*$---otherwise $\{k'-1, k', k'+1\}$ would all be holes.)
    Replace $k$ throughout $T^*$ by $k'$.
    As $k'$ and $k$ are adjacent keys and $k'$ does not occur in $T^*$,
    the resulting tree $\barT$ still solves $(I_{15}, 2)$,
    and $\barT$ costs less than $T^*$ (as $\barT$ uses the weight-0 key $k'$ instead of the weight-5 key $k$).
    This contradicts the optimality of $T^*$.     

  \end{description}

  By (P3), $T^*$ has seven positive-weight queries.
 Using condition (P2) and left-right symmetry of subproblem $(I_{15}, 2)$, we can assume
	  that the left subtree of $T^*$ has at least four of the seven.
	(Note that ``flipping'' the tree, namely replacing each key $k$ by $16-k$ and swapping the 
	yes and no-subtrees, would map each inequality
	comparison $\compnode{<}{k}$ to $\compnode{\le}{16-k}$, while our model uses only strict inequalities.
	However, this latter comparison is equivalent to $\compnode{<}{17-k}$.)
  Let $T'$ be the left subtree. Denote the S-subproblem that $T'$ solves by $(I_\ell, h')$.
  To prove the lemma, assume for contradiction that $T'$ is optimal for its S-subproblem,
  and proceed by cases:
  
  \smallskip
  
  \begin{outercases}\setlength{\itemsep}{0.75em}
  \item[1]\emph{$T'$ has four positive-weight queries.}
    That is, $T'$ solves an S-subproblem $(I_\ell, h')$ where $\nofposvalues{\ell}-h' = 4$.
    By Lemma~\ref{lemma: T 4}, $T'$ has cost 49 and weight 22.
    The right subtree $T''$ of $T^*$ has the three remaining positive-weight queries,
    the leftmost two of which are separated in $T''$ by a zero-weight query (using (P3)).
    By Lemma~\ref{lemma: depth}\,(ii),
    $T''$ has cost at least $5\cdot e_3 = 30$ and weight at least 15.
    The cost of $T^*$ is its weight plus the costs of $T'$ and $T''$.
    By the above observations, this is at least $(22+15)+49+30 = 116$, contradicting (P1).

  \item[2]\emph{$T'$ has five positive-weight queries.}
    That is, $T'$ solves an S-subproblem $(I_\ell, h')$ where $\nofposvalues{\ell}-h' = 5$.
    By Lemma~\ref{lemma: T 5}, $T'$ has cost 69 and weight 27.
    The right subtree $T''$ of $T^*$ has the two other positive-weight queries,
    which have total depth at least $1+1=2$ in $T''$, and each has weight at least 5.
    So $T''$ has cost, and weight, at least $5\cdot 2 = 10$.
    The cost of $T^*$ is its weight plus the costs of $T'$ and $T''$.
    By the above observations, this is at least $(27+10) + 69 + 10 = 116$, contradicting (P1).

  \item[3]\emph{$T'$ has six or seven positive-weight queries.}
    Let set $S$ consist of just the first six of these queries.
    Since $T'$ is the left subtree of $T^*$ (which has seven positive-weight queries)
    $S$ does not contain the last key, 15.
    So (using (P3)) all queries in $S$, except possibly $\{1, 2\}$, are separated by weight-zero queries in $T'$.
    By Lemma~\ref{lemma: depth}\,(ii), $T'$ has cost at least $5\cdot e_6 = 90$.
    The cost of $T^*$ is its weight (at least $7\cdot 5 = 35$),
    plus the cost of its left and right subtrees (at least 90, counting $T'$ alone).
    So $T^*$ costs at least $35+90=125$, contradicting (P1).
  \end{outercases}

  This proves the lemma and Theorem~\ref{thm: no optimal substructure spuler}.
  \hfill\qed
  \vspace*{-1.1em}

  ~
\end{proof}


Finally we prove Theorem~\ref{thm: spuler flaw}.

\begin{proof}[Theorem~\ref{thm: spuler flaw}]
  \begin{linenomath*}
    Consider any execution of Spuler's algorithm on the S-subproblem
    $(I, h)$  from Theorem~\ref{thm: no optimal substructure spuler}, breaking ties arbitrarily.
    Let $T$ be the tree it computes for that S-subproblem.
    By Theorem~\ref{thm: no optimal substructure spuler},  either $T$ is not optimal for $(I,h)$,
    or some subtree $T'$ of $T$ is not optimal for its S-subproblem $(I', h')$.
    So Spuler's algorithm must compute a non-optimal solution to at least one S-subproblem.
    \hfill\qed
    \vspace*{-1.1em}

  ~
\end{linenomath*}
\end{proof}

In fact, for this instance $(I, h)$,
Spuler's algorithm (as implemented via the Python code in Appendix~\ref{sec: code spuler})
computes a non-optimal tree of cost 116,
such as \sfigtree F a in Fig.~\ref{fig: spuler counter F}.
By inspection, tree \sfigtree F b in that figure costs 115, so \sfigtree F a is not optimal.

 
\paragraph{Discussion.}
As mentioned earlier, this counter-example is just for a subproblem.
This subproblem has $h=2$ holes, so it does not
represent a complete instance of \twoWCST
for which Spuler's algorithm would give an incorrect \emph{final} result.
However, this counter-example does demonstrate that Spuler's algorithm
solves some \emph{subproblems} incorrectly,
so that the recurrence relation underlying its dynamic program is incorrect.
At a minimum, this suggests that any proof of correctness for Spuler's algorithm
would require a more delicate approach. Anderson~{\etal}~\cite{Anderson2002}
establish some conditions on the weights of equality-test keys in optimal trees.
  It may be possible to leverage the bounds from~\cite{Anderson2002} to show that
  bad subproblems---those that are not solved correctly by the algorithm---never appear as subproblems of an optimal complete tree.
  For example, per Anderson {\etals} Theorem 5 
  for any equality-test node in any optimal tree, the weight of the node's key
  must be at least one quarter of the total weight of the keys that reach the node.
  Hence, if a subproblem $(I', h')$ is solved by some subtree $T'$
  of an optimal tree $T^*$, then each hole key in $T'$
  must have weight at least one third of the total weight of the queries in $T'$.
  This implies that the subproblem $(I_{15}, 2)$ in the proof of Theorem~\ref{thm: spuler flaw}
  cannot actually occur in any optimal tree for $(I_{15}, 0)$.

While the question of correctness of Spuler's algorithm is somewhat
intriguing, it should be noted that showing its correctness will 
not improve known complexity bounds for \twoWCST, as
there are faster \twoWCST algorithms that are known to be 
correct~\cite{Anderson2002,chrobak_etal_isaac_2015}.


\subsection{Proof of Lemma~\ref{lemma: depth}.}
Here is the promised proof of Lemma~\ref{lemma: depth}.

\begin{proof}[Lemma~\ref{lemma: depth}]
    Recall that $T$ is a tree for some S-subproblem 
    and $Q$ is a subset of the queries in $T$, with $m = |Q|$.
    \smallskip
    
    \noindent\emph{Part (i).} Assume that $Q$ is separated. 
    Our goal is to show that the total depth in $T$ of queries in $Q$ is at least $d_m$,
    as defined before Lemma~\ref{lemma: depth}.
    It is convenient to recast the problem as follows.
    Change the weight of each query in $Q$ to 1. Change the weight of each query not in $Q$ to 0.
	We will refer to the resulting cost of a tree as
	\emph{modified cost}.
    Now we need to show that the modified cost of $T$ is at least $d_m$.
    The proof is by induction on $m$.

The base cases (when $m = 1,2$) are easily verified, so consider the inductive step,
for some given $m \ge 3$. We assume that $T$ and $Q$ are chosen to
minimize the modified cost of $T$, subject to $|Q|=m$.  
Call this the \emph{minimality assumption}.

Suppose $T$ does an inequality test at the root.
Let $T_1$ and $T_2$ be the left and right subtrees of $T$,
and for $a\in\braced{1,2}$ let $Q_a\subseteq Q$ contain the queries in $Q$ that fall in $T_a$.
Let $i = |Q_1|$, so that $|Q_2| = m-i$.
For $a\in\braced{1,2}$, query set $Q_a$ is $T_a$-separated.
By the minimality assumption, $0\not\in\{i, m-i\}$.
The modified cost of $T$ is its weight ($m$),
plus the modified costs of $T_1$ and $T_2$.
By the inductive assumption, this is at least $m+d_{i} + d_{m-i} \ge d_m$, as desired.

Suppose $T$ does an equality test at the root. The minimality assumption implies that
the equality-test key has non-zero (modified) weight.
 (This follows via the argument given for Property (P2) in the proof of Lemma~\ref{lemma: obstacle},
 using Anderson \etals Theorem~5 or Corollary~3.)
So the equality-test key is in $Q$.  Let $T_1$ be the no-subtree of $T$ and
 let $Q_1\subseteq Q$ contain the queries in $Q$ that fall in $T_1$;
so we have $|Q_1|=m-1$. Set $Q_1$ is $T_1$-separated, so by the inductive assumption,
$T_1$ has modified cost at least $d_{m-1}$.
So the modified cost of $T$ is at least $m+ d_{m-1} = m + d_1 + d_{m-1} \ge d_m$, as desired.

\smallskip

\begin{linenomath*}
\noindent\emph{Part (ii).} The proof of Part (ii) follows the same inductive argument as above.
The base cases for $m=1,2$ are trivial. The verification of the base case
for $m=3$ is by straightforward case analysis. 
In the inductive step, the only significant difference is in the case
when $T$ does an inequality test at the root. Since $Q$ is now only nearly separated, 
$Q_1$ will be $T_1$-separated while $Q_2$ will be nearly $T_2$-separated
(or vice versa), giving us that the modified cost of 
$T$ is at least $m+d_{i} + e_{m-i} \ge e_m$.
  \hfill\qed
  \vspace*{-1.1em}

~
\end{linenomath*}
\end{proof}


%% file: TREE_Spuler_example.tex
\bracketset{action character=@}

{\renewcommand{\subtreecolor}{black!15}
  \begin{yntree}
    for tree={s sep=1.5em}, 
    [ {$\query = 3$}, equalitytest,
      [ {$3$}, @\lab{0.6} ]
      [ {$\query < 4$}, 
        [ {$\query < 2$}, 
          [ {$1$}, @\lab{0.1} ]
          [ {$2$}, @\lab{0.1} ]
          ]
        [ {$\query < 6$}, 
          [ {$\query = 5$}, equalitytest,
            [ {$5$}, @\lab{0.05} ]
            [{4},  @\lab{0.05} ]
            ]
          [ {$6$}, @\lab{0.1} ]
          ]
        ]
      ]
  \end{yntree}
}


%% file: TREE_Spuler_counter_example_4.tex
\bracketset{action character=@}

{\renewcommand{\subtreecolor}{black!15}
\begin{yntrees}
  for tree={s sep=1.5em}, 
[, phantom, ,
  [{$\sfigtree 4 a$}, nodraw, below=4ex ]
  [, nodraw
          [{~$N$\,~}, edge=dashed, edge label={}, s sep=3em,
          [ {$\query < 3$}, edge=dashed, edge label={},
                   tikz={
        \begin{scope}[on background layer]
          \node
          [draw, dashed, fill=gray, fill opacity=.03, 
              fit=()(!11)(!2)(!222), 
              semicircle, rounded corners=2em, 
              inner sep=0em,
              yshift=-1.85ex, 
              xshift=1.1ex, 
              xscale=0.6,
              yscale=1.09
          ]
          {};
        \end{scope}
      }
  [ {$\query = 1$}, equalitytest,
    [ {$1$}, @\lab{\mathbf 7} ]
    [ {$2$}, @\lab{5} ]
  ]
  [ {$\query = 4$}, equalitytest,
    [ {$4$}, @\lab{5} ]
    [ {$\query = 6$}, equalitytest,
      [ {$6$}, @\lab{5} ]
        [{},  @\lab{0}, subtree ]
    ]
  ]
  ]
        [, opacity=0, edge=dotted, edge label={}]
   ]]
   [, phantom [, phantom[, phantom [, phantom [, phantom]]]]]
  [{$\sfigtree 4 b$}, nodraw, xshift=8em, below=4ex ]
  [, nodraw
            [{~$N$\,~}, edge=dashed, edge label={}, s sep=3em,
[ {$\query < 5$}, edge=dashed, edge label={},
                   tikz={
        \begin{scope}[on background layer]
          \node
          [draw, dashed, fill=gray, fill opacity=.03, 
              fit=()(!11)(!121)(!2)(!222), 
              semicircle, rounded corners=2.3em, 
              inner sep=0em,
              yshift=-2.2ex, 
              xshift=.5ex, 
              xscale=0.67,
              yscale=1.07
          ]
          {};
        \end{scope}
      }
  [ {$\query = 2$}, equalitytest,
    [ {$2$}, @\lab{5} ]
    [ {$\query = 4$}, equalitytest,
      [ {$4$}, @\lab{5} ]
      [{},  @\lab{0}, subtree ]
    ]
  ]
  [ {$\query = 6$}, equalitytest, 
   [ {$6$}, @\lab{5} ]
  [ {$\query = 8$}, equalitytest,
    [ {$8$}, @\lab{5} ]
      [{},  @\lab{0}, subtree ]
    ]
  ]]
          [, opacity=0, edge=dotted, edge label={}]
          ]]
   [, phantom [, phantom[, phantom [, phantom [, phantom]]]]]
  [{$\sfigtree 4 c$}, nodraw, xshift=-1em, below=5ex ]
  [, nodraw
          [{~$N$\,~}, edge=dashed, edge label={}, 
          [ {$\query = 2$}, equalitytest, edge=dashed, edge label={},
  tikz={
        \begin{scope}[on background layer]
          \node
          [draw, dashed, fill=gray, fill opacity=.03, 
              fit=()(!1)(!2222),
              trapezium, trapezium angle=85,
              rounded corners=1.7em, 
              inner sep=0em,
              yshift=-1ex, 
              xshift=0ex, 
              xscale=1.1,
              yscale=1.15
          ]
          {};
        \end{scope}
      }
  [ {$2$}, @\lab{5} ]
  [ {$\query = 4$}, equalitytest,
    [ {$4$}, @\lab{5} ]
    [ {$\query = 6$}, equalitytest,
      [ {$6$}, @\lab{5} ]
      [ {$\query = 8$}, equalitytest,
        [ {$8$}, @\lab{5} ]
        [{},  @\lab{0}, subtree ]
      ]
    ]
   ]
   ]
      [, opacity=0, edge=dotted, edge label={}]
      ]
    ]
  ]
\end{yntrees}
}


%% file: TREE_Spuler_counter_example_5.tex
\bracketset{action character=@}

{\renewcommand{\subtreecolor}{black!15}
\begin{yntrees}
  for tree={s sep=1.5em}, 
[, phantom,,
  [{$\sfigtree 5 a$}, nodraw, below=4ex ]
  [, nodraw
          [{~$N$\,~}, edge=dashed, edge label={}, s sep=3em,
          [ {$\query < 3$}, edge=dashed, edge label={},
                   tikz={
        \begin{scope}[on background layer]
          \node
          [draw, dashed, fill=gray, fill opacity=.03, 
              fit=()(!11)(!2)(!2222), 
              semicircle, rounded corners=2em, 
              inner sep=0em,
              yshift=-2ex, 
              xshift=1.1ex, 
              xscale=0.56,
              yscale=1.09
          ]
          {};
        \end{scope}
      }
  [ {$\query = 1$}, equalitytest,
    [ {$1$}, @\lab{7} ]
    [ {$2$}, @\lab{5} ]
  ]
  [ {$\query = 4$}, equalitytest,
    [ {$4$}, @\lab{5} ]
    [ {$\query = 6$}, equalitytest,
      [ {$6$}, @\lab{5} ]
      [ {$\query = 8$}, equalitytest,
        [ {$8$}, @\lab{5} ]
        [{},  @\lab{0}, subtree ]
      ]
    ]
  ]
  ]
        [, opacity=0, edge=dotted, edge label={}]
   ]]
   [, phantom [, phantom[, phantom [, phantom [, phantom]]]]]
   [, phantom [, phantom[, phantom [, phantom [, phantom]]]]]
  [{$\sfigtree 5 b$}, nodraw, xshift=8em, below=4ex ]
  [, nodraw
            [{~$N$\,~}, edge=dashed, edge label={}, s sep=3em,
[ {$\query < 7$}, edge=dashed, edge label={},
                   tikz={
        \begin{scope}[on background layer]
          \node
          [draw, dashed, fill=gray, fill opacity=.03, 
              fit=()(!11)(!2)(!2222), 
              semicircle, rounded corners=2em, 
              inner sep=0em,
              yshift=-2.5ex, 
              xshift=1.3ex, 
              xscale=0.65,
              yscale=1.04
          ]
          {};
        \end{scope}
      }
  [ {$\query = 4$}, equalitytest,
    [ {$4$}, @\lab{5} ]
    [ {$\query = 6$}, equalitytest,
      [ {$6$}, @\lab{5} ]
      [{},  @\lab{0}, subtree ]
    ]
  ]
  [ {$\query = 8$}, equalitytest, 
   [ {$8$}, @\lab{5} ]
  [ {$\query = 10$}, equalitytest,
    [ {$10$}, @\lab{5} ]
    [ {$\query = 12$}, equalitytest,
      [ {$12$}, @\lab{5} ]
      [{},  @\lab{0}, subtree ]
    ]
  ]
  ]]
          [, opacity=0, edge=dotted, edge label={}]
    ]]
  ]
\end{yntrees}
}

%% file: TREE_Spuler_counter_example_full.tex
\bracketset{action character=@}

\begin{yntrees}
  for tree={s sep=1.5em},,,
[, phantom,
  [{$\sfigtree F a$}, nodraw, xshift=4em, below=0.5ex ]
  [, nodraw
[ {$\query < 9$}, edge=dashed, edge label={}, s sep=2.5em,
  [ {$\query < 3$}, tikz={
        \begin{scope}[on background layer]
          \node
          [draw, gray, dashed, fill opacity=.0, 
              fit=()(!11)(!2)(!222), 
              semicircle, rounded corners=2em, 
              inner sep=0em,
              yshift=-1.85ex, 
              xshift=1ex, 
              xscale=0.58,
              yscale=1.1
          ]
          {};
        \end{scope}
      }
    [ {$\query = 1$}, equalitytest,
      [ {$1$}, @\lab{7} ]
      [ {$2$}, @\lab{5} ]
    ]
    [ {$\query = 4$}, equalitytest,
      [ {$4$}, @\lab{5} ]
      [ {$\query = 6$}, equalitytest,
        [ {$6$}, @\lab{5} ]
        [{},  @\lab{0}, subtree ]
        ]
    ]
  ]
  [ {$\query = 10$}, equalitytest,
    [ {$10$}, @\lab{5} ]
    [ {$\query = 12$}, equalitytest,
      [ {$12$}, @\lab{5} ]
      [ {$\query = 14$}, equalitytest,
        [ {$14$}, @\lab{5} ]
        [{},  @\lab{0}, subtree ]
      ]
    ]
  ]
  ]
    ]
[, phantom [, phantom[, phantom]]]
[{$\sfigtree F b$}, nodraw, xshift=8em, below=0.5ex ]
  [, nodraw
[ {$\query < 9$},  edge=dashed, edge label={},
          [ {$\query = 2$}, equalitytest, 
  tikz={
        \begin{scope}[on background layer]
          \node
          [draw, gray, dashed, fill opacity=.0, 
              fit=()(!1)(!2222),
              trapezium, trapezium angle=85,
              rounded corners=1.7em, 
              inner sep=0em,
              yshift=-1ex, 
              xshift=0ex, 
              xscale=1.1,
              yscale=1.15
          ]
          {};
        \end{scope}
      }
  [ {$2$}, @\lab{5} ]
  [ {$\query = 4$}, equalitytest,
    [ {$4$}, @\lab{5} ]
    [ {$\query = 6$}, equalitytest,
      [ {$6$}, @\lab{5} ]
      [ {$\query = 8$}, equalitytest,
        [ {$8$}, @\lab{5} ]
        [{},  @\lab{0}, subtree ]
      ]
    ]
   ]
   ]
  [ {$\query = 10$}, equalitytest,
      [ {$10$}, @\lab{5} ]
      [ {$\query = 12$}, equalitytest,
        [ {$12$}, @\lab{5} ]
        [ {$\query = 14$}, equalitytest,
          [ {$14$}, @\lab{5} ]
          [{},  @\lab{0}, subtree ]
        ]
      ]
    ]
]
  ]]
\end{yntrees}


%% file: 5_python-code_v3.tex
\section{--- Python code for Huang and Wong's \gbsplit algorithm}\label{sec: code huang}



\begin{lstlisting}
#!/usr/bin/env python3.9
from collections import namedtuple
from functools import lru_cache

memoize = lru_cache(maxsize=None); Tree = namedtuple('Tree', 'cost weight holes')

def all_mins(items, key=lambda x: x):
    ''' return all minimums in set of items '''
    return {x for m in [min((key(x) for x in items), default=None)] for x in items if key(x) == m}

def size(i, j, h):
    return j - i + 1 - h

def huang1984(weights, holes=0):
    '''Simulates Huang and Wong's GBSPLIT algorithm (1984), trying all ways of breaking ties.'''

    wts = [weights[k] for k in sorted(weights.keys())]

    @memoize
    def trees(i, j, h):
        return ({ Tree(cost=0, weight=0, holes=frozenset(range(i, j+1))) } if size(i, j, h) == 0 else
                all_mins({
                    Tree(weight + left.cost + right.cost, weight, holes - {eq_key})
                    for k in range(i, j+2) for h_l in range(h+2)
                    if size(i, k-1, h_l) >= 0 and size(k, j, h_r := h-h_l+1) >= 0
                    for left in trees(i, k-1, h_l) for right in trees(k, j, h_r)
                    for holes in [left.holes | right.holes]
                    for eq_key in all_mins(holes, key=lambda H: wts[H])
                    for weight in [left.weight + right.weight + wts[eq_key]]
                }, key=lambda t: t.cost))

    return trees(0, len(weights)-1, holes)

# The instance (K, p) from the proof of Theorem 1:
weights = dict(B4=20,
               A3=20, V3=20,
               A2=20, F2=20, T2=20, X2=20,
               A1=20, D1=22, F1=20, Q1=20, S1=20, U1=20, W1=20, Y1=20,
               B0=10, C0= 5, D0=10, E0=10, N0=10, P0=10, Q0=10, R0=10,
               S0=10, T0=10, U0=10, V0=10, W0=10, X0=10, Y0=10, Z0=10)

solutions = huang1984(weights); assert len(solutions) == 1 and solutions.pop().cost == 1763

# Increasing a weight _lowers_ the cost computed by the algorithm.
weights['D1'] += 0.99; solutions = huang1984(weights);
assert len(solutions) == 1 and solutions.pop().cost < 1763

# The instance (I9, 2) from the proof of Theorem 2:
weights = dict(A1=20, A2=20, A3=20, B0=10, B4=20, C0= 5, D0=10, D1=22, E0=10)
assert huang1984(weights, 2).pop().cost == 209
\end{lstlisting}



\newpage
\section{--- Python code for Spuler's \twoWCST algorithm}\label{sec: code spuler}


\begin{lstlisting}
#!/usr/bin/env python3.9
from collections import namedtuple
from functools import lru_cache

memoize = lru_cache(maxsize=None); Tree = namedtuple('Tree', 'cost weight holes')

def all_mins(items, key=lambda x: x):
    return {x for m in [min((key(x) for x in items), default=None)] for x in items if key(x) == m}

def size(i, j, h):
    return j - i + 1 - h

def spuler1994(wts, n_holes):
    '''Simulates Spuler's 2WCST algorithm (1994), trying all ways of breaking ties'''

    @memoize
    def trees(i, j, h):
        interval = frozenset(range(i, j+1))
        return ({
            Tree(cost=0, weight=wts[k], holes=interval - {k})
            for k in all_mins(interval, key=lambda k: wts[k])
        } if size(i, j, h) == 1 else
        all_mins({         # possible equality tests
            Tree(cost=wts[eq_key] + right.weight + right.cost,
                 weight=wts[eq_key] + right.weight,
                 holes=right.holes - {eq_key})
            for right in trees(i, j, h+1) for eq_key in all_mins(right.holes, key=lambda k: wts[k])
        } | {                   # possible inequality tests
            Tree(cost=left.weight + right.weight + left.cost + right.cost,
                 weight=left.weight + right.weight,
                 holes=left.holes | right.holes)
            for k in interval for h_l in range(h+1) if size(i, k-1, h_l) >= 1 and size(k, j, h_r := h-h_l) >= 1
            for left in trees(i, k-1, h_l) for right in trees(k, j, h_r)
        }, key=lambda t: t.cost))

    return trees(0, len(wts)-1, n_holes)

# The instance (K, p) and subproblem (I_15, 2) from the proof of Theorem 4:
weights = [7, 5, 0, 5, 0, 5, 0, 5, 0, 5, 0, 5, 0, 5, 7];
assert spuler1994(weights, 2).pop().cost == 116

# Increasing some weights _lowers_ the cost computed by the algorithm:
weights = [9, 5, 0, 5, 0, 5, 0, 5, 0, 5, 0, 5, 0, 5, 9];
solutions = spuler1994(weights, 2)
assert len(solutions) == 1 and solutions.pop().cost == 115
\end{lstlisting}


